\documentclass{article}

\usepackage{amsmath,amsthm,amssymb}
\usepackage{xspace}
\usepackage{graphicx}
\usepackage{xcolor}
\usepackage{hyperref}
\usepackage{multirow}
\usepackage{tikz}
\usepackage[utf8]{inputenc}
\usepackage{enumerate}
\usepackage{fullpage}

\newtheorem{theorem}{Theorem}
\newtheorem{lemma}[theorem]{Lemma}

\theoremstyle{definition}

\newtheorem*{lemma*}{Lemma}

\newcommand{\smallqed}{\hfill {\tiny ($\diamond$)} \medskip}

\newcommand{\2}{\vspace{0.2cm}}

\newcommand{\AY}[1]{{#1}}
\newcommand{\GG}[1]{{#1}}
\newcommand{\AZ}[1]{{#1}}

\newcommand{\add}[1]{a^*_{#1}}
\newcommand{\del}[1]{d^*_{#1}}

\newcommand{\addX}[1]{a^+_{#1}}
\newcommand{\delX}[1]{d^+_{#1}}

\title{(1,1)-Cluster Editing is Polynomial-time Solvable}

\author{Gregory Gutin\thanks{Department of Computer Science. Royal Holloway University of London. {\tt g.gutin@rhul.ac.uk}.} \and Anders Yeo \thanks {Department of Mathematics and Computer Science, University of Southern Denmark, Denmark and Department of Pure and Applied Mathematics,
University of Johannesburg, South Africa. {\tt andersyeo@gmail.com.}}}

\begin{document}

\maketitle

\begin{abstract}
A graph $H$ is a {\em clique graph} if $H$ is a vertex-disjoin union of cliques.  Abu-Khzam (2017) introduced the  $(a,d)$-{Cluster Editing} problem, where for fixed natural numbers $a,d$, given 
a graph $G$ and vertex-weights $\add{}:\ V(G)\rightarrow \{0,1,\dots, a\}$ and  $\del{}:\ V(G)\rightarrow \{0,1,\dots, d\}$, 
we are to decide whether $G$ can be turned into a cluster graph 
by deleting at most $\del{}(v)$ edges incident to every $v\in V(G)$ and adding at most $\add{}(v)$ edges incident to every $v\in V(G)$. Results by Komusiewicz and Uhlmann (2012) and Abu-Khzam (2017) provided a dichotomy of complexity (in P or NP-complete) of $(a,d)$-{Cluster Editing} for all pairs $a,d$ apart from $a=d=1.$ Abu-Khzam (2017) conjectured that $(1,1)$-{Cluster Editing} is in P. 
We resolve Abu-Khzam's conjecture in affirmative by (i) providing a series of five polynomial-time reductions to $C_3$-free and $C_4$-free graphs of maximum degree at most 3, and (ii)  designing a polynomial-time algorithm for solving $(1,1)$-{Cluster Editing} on $C_3$-free and $C_4$-free graphs of maximum degree at most 3.
\end{abstract}

\section{Introduction}

A graph $H$ is a {\em clique graph} if $H$ is a vertex-disjoin union of cliques. Addition or deletion of an edge to a graph is called an {\em edge edition}.
Given a graph $G$ and an integer $k \ge 0,$ the Cluster Editing problem asks whether $k$ or less edge editions can transform $G$ into a clique graph. 
Cluster Editing is {\sf NP}-Complete \cite{KrivanekM86,ShamirST04}, but it is fixed-parameter tractable if parameterized by $k$ \cite{Cai96}. The parameterization 
has received considerable attention \cite{Bocker12,BockerBBT09,BockerBK11,CaoC12,ChenM12,GrammGHN05,Guo09}. In particular,  an $O^*(1.618^k)$-time algorithm was designed in \cite{Bocker12} and a kernel with at most $2k$ vertices in \cite{ChenM12}. Unfortunately, $k$ is often not  small for real-world instances. For example, in a protein similarity data set that has been frequently used for evaluating Cluster Editing algorithms, the instances with $n \ge 30,$ where $n$ is the number of vertices, have an average number of edge editions between $2n$ and $4n$ \cite{BockerBBT09}. This has led to introduction of \AZ{multi-parameter} parameterizations in \cite{Abu-Khzam2017,KomusiewiczU2012}, \AZ{which turned out to be very useful in solving cluster editing problems in practice, see e.g. \cite{Barr0AC19,BarrSATY20,ShawBA22}.}

We will consider the following variation of {\sc Cluster Editing} problem introduced by Abu-Khzam \cite{Abu-Khzam2017}.
In $(a,d)$-{\sc Cluster Editing}, for fixed natural numbers $a$ and $d$,
given a graph $G$ and vertex-weights $\add{}:\ V(G)\rightarrow \{0,1,\dots, a\}$ and  $\del{}:\ V(G)\rightarrow \{0,1,\dots, d\}$, 
decide whether $G$ can be turned into a cluster graph 
by deleting at most $\del{}(v)$ edges incident to every $v\in V(G)$ and adding at most $\add{}(v)$ edges incident to every  $v\in V(G)$. Note that
there is no upper bound on the total number of additions and deletions, see Paragraph 7 of Section 2 and Paragraph 1 of Section 5.1 in \cite{Abu-Khzam2017}.
$(a,d)$-{\sc Cluster Editing} is similar to a version of  {\sc Cluster Editing} studied by Komusiewicz and Uhlmann \cite{KomusiewiczU2012}.
In particular, they proved that $(0,d)$-{\sc Cluster Editing} is {\sf NP}-hard for every fixed integer $d\ge 2.$

Abu-Khzam \cite{Abu-Khzam2017} proved additional results on $(a,d)$-{\sc Cluster Editing} and
established computational complexity of $(a,d)$-{\sc Cluster Editing} for every value of 
$a$ and $d$ apart from $a=d=1$: $(a,d)$-{\sc Cluster Editing} is in {\sf P} when $(a,d)\in \{(a,0), (0,1):\  a\in \{0,1,2\}\}$ and {\sf NP}-complete, otherwise, apart from possibly the case of $a=d=1$. 
In fact, Abu-Khzam \cite{Abu-Khzam2017} conjectured that $(1,1)$-{\sc Cluster Editing} is polynomial-time solvable. 

In this paper we resolve Abu-Khzam's conjecture in affirmative using \AZ{the following two-stage approach:
\begin{description}
\item[Stage 1] A polynomial-time self-reduction to $(1,1)$-{\sc Cluster Editing} on a special class of $\{C_3,C_4\}$-free graphs of maximum degree at most 3.
\item[Stage 2] Solving $(1,1)$-{\sc Cluster Editing} on this special graph class.
\end{description}
}

\AZ{Note that Stage 1 is actually a series of five polynomial-time reductions and is technically significantly harder than Stage 2. Thus the remaining sections of the paper are organized as follows. In Section \ref{sec:scheme} we introduce and briefly overview the series of reductions of Stage 1 leaving all the details to Section \ref{sec:5}. Stage 2 is described in Section \ref{sec:pta}. We conclude the paper with a discussion in Section \ref{sec:disc}. 
}

\vspace{2mm} 

We conclude this section by additional terminology and notation.

An edge $uv$ of $G$ is {\em deletable (in $G$)} if $\del{}(u)=\del{}(v)=1$ and {\em non-deletable}, otherwise. 
An edge $uv$ of the complement $\overline{G}$ of $G$ is {\em addable (to $G$)} if $\add{}(u)=\add{}(v)=1$ and {\em non-addable}, otherwise. 
Now we reformulate $(1,1)$-{\sc Cluster Editing} as follows:
Given a graph $G$, decide whether $G$ has a matching $D$ of edges deletable in $G$ and $\overline{G}$  has a matching $A$ of edges addable to $G$ such that $G-D+A$ is a clique graph. We call it the {\em matching} formulation of $(1,1)$-{\sc Cluster Editing} and the desirable clique graph $G-D+A$ a {\em solution}. Note that we will often write a solution in the form $G-D+A$, which defines the matchings $D$ and $A$ rather than the resulting cliques. 

We will denote by $P_n$ ($C_n$, respectively) a path (cycle, respectively) on $n$ vertices, by $K_n$ a complete graph on $n$ vertices, and by $K_{n,m}$ a complete bipartite graph with partite sets of sizes $n$ and $m.$ 
\AZ{We say that a graph $G$ is {\em $C_3$-free} if it does not contain a cycle $C_3$ as a subgraph
and is {\em $C_4$-free} if it does not contain a cycle $C_4$ as a subgraph. A graph $G$ is $\{C_3,C_4\}$-free if it is both $C_3$-free and $C_4$-free.}
For a positive integer $p$, let $[p]=\{1,2,\dots , p\}.$

\AZ{ \section{Scheme of series of reductions to special $\{C_3,C_4\}$-free graphs of maximum degree at most 3}\label{sec:scheme}

For two sets ${\cal G}$ and ${\cal H}$ of graphs, we say that $(1,1)$-{\sc Cluster Editing} can be {\em reduced} 
from ${\cal G}$ to ${\cal H}$ if for all $G\in {\cal G}$ and its vertex-weights $\add{},\del{}$,
we can either decide in polynomial time in $|V(G)|$ whether $(G,\add{},\del{})$ is a YES-instance or NO-instance, or construct $H \in {\cal H}$ with vertex-weights $\addX{},\delX{}$ such that $(G,\add{},\del{})$ is a YES-instance of $(1,1)$-{\sc Cluster Editing} if and only if $(H,\addX{},\delX{})$ is a YES-instance of $(1,1)$-{\sc Cluster Editing}. 
Note that as the reduction is a polynomial-time algorithm, the order of $H$ is also bounded above by a polynomial in the order of $G$.
When the vertex-weights are clear from the context, we will often say that $G$ is an instance of the problem (i.e., of $(1,1)$-{\sc Cluster Editing}). 


Let ${\cal G}_1$ denote all graphs and let 
${\cal G}_2$ denote all $C_3$-free graphs of maximum degree at most 3. The first in our series of five reductions is as follows.  

\begin{lemma}\label{lem:tf}
 $(1,1)$-{\sc Cluster Editing} can be reduced from ${\cal G}_1$ to ${\cal G}_2$.
\end{lemma} 

Let ${\cal G}_3$ denote all $C_3$-free graphs of maximum degree at most 3 which contain no $4$-cycle, all of whose vertices have degree 3.
 
\begin{lemma}\label{lem:noC4with4deg3}
 $(1,1)$-{\sc Cluster Editing} can be reduced from ${\cal G}_2$ to ${\cal G}_3$ in polynomial time.
\end{lemma}

Let ${\cal G}_4$ denote all $C_3$-free graphs of maximum degree at most 3 which contain no $4$-cycle, with at least three vertices of degree three.
 
\begin{lemma}\label{lem:noC4with3deg3}
 $(1,1)$-{\sc Cluster Editing} can be reduced from ${\cal G}_3$ to ${\cal G}_4$ in polynomial time.
\end{lemma}

Let ${\cal G}_5$ denote all $\{C_3,C_4\}$-free graphs of maximum degree at most 3.
 
\begin{lemma}\label{lem:noC4}
 $(1,1)$-{\sc Cluster Editing} can be reduced from ${\cal G}_4$ to ${\cal G}_5$ in polynomial time.
\end{lemma}

Let ${\cal G}_6$ denote all $\{C_3,C_4\}$-free graphs of maximum degree at most 3, such that the following holds.

\begin{itemize}
\item All vertices $v$ of degree at least 2 have $\del{}(v)=1$.
\item No vertex $v$ of degree 3 is adjacent to a vertex $w$ with $\add{}(w)=0$.
\item The vertices of degree 3 form an independent set.
\end{itemize}

Graphs in ${\cal G}_6$ are called {\em special} $\{C_3,C_4\}$-free graphs of maximum degree at most 3.

\begin{lemma}\label{lem:dZero}
 $(1,1)$-{\sc Cluster Editing} can be reduced from ${\cal G}_5$ to ${\cal G}_6$ in polynomial time.
 \end{lemma}
}
\section{Polynomial-time algorithm for  \AZ{special} $\{C_3,C_4\}$-free graphs of maximum degree at most 3}\label{sec:pta}

Lemma \ref{lem:posI} of Section \ref{sec:5} implies the following lemma. \AZ{In fact, it is easy to prove Lemma \ref{lem:posII} directly.} 

 \begin{lemma}\label{lem:posII}
 Let $G$ be a $\{C_3,C_4\}$-free graph with maximum degree at most 3. 
Then $(G,\add{},\del{})$ is a YES-instance of $(1,1)$-{\sc Cluster Editing} if and only if there is a matching $D$ of deletable edges from $G$ such that every connected 
component $C$ of $G-D$ is isomorphic to $P_1$ or $P_2$ or $P_3$ and if $C\cong P_3$ then the edge of $\overline{G}$ between the end-vertices of $C$ is addable
to $G$.
 \end{lemma}

\AZ{A matching $M$ in a graph $G$ {\em covers} a vertex $v$ in $G$ if $v$ is an end-point of an edge in $M.$}
In the proof of the next lemma, we will use the well-known result that given a graph $G$ in polynomial time either a perfect matching can be found in $G$ 
or we \AY{can} decide that $G$ has no perfect matching. 

\begin{lemma}\label{thm:EdmondBlossum}
Let $G$ be a graph and let $Y$ and $Z$ be disjoint sets of vertices in $G$.
In polynomial time, we can decide whether $G$ contains a matching covering \AZ{every vertex in $Y$, but no vertex in $Z$}.
\end{lemma}
\begin{proof}
Let $G'=G-Z$. Observe that $G$ contains a matching covering \AZ{every vertex of $Y$ and no vertex $Z$} if and only if $G'$ contains a matching covering $Y$. 
Let $U= V(G')-Y$ and if $G' $ has odd number of vertices then add an extra (isolated) vertex to $G'$ which is not in $Y$ (and therefore will be added to $U$). Finally add edges between all pairs of non-adjacent vertices in $U$. Let the resulting graph be $G''.$  Now we will show that $G''$ has a perfect matching if and only if there is a matching in $G '$ covering $Y.$ If $M''$ is a perfect matching in $G''$ then by deleting all edges with both end-points in $U$ we get the desired matching in $G'.$  And if $M'$ is the desired matching in $G'$ then adding as many edges as possible with both end-points in $U$ gives us a perfect matching in $G''.$
\end{proof}

\begin{lemma}\label{lem:polAlg}
 $(1,1)$-{\sc Cluster Editing} can be solved in polynomial time for all graphs in ${\cal G}_6$ in polynomial time.
\end{lemma}
\begin{proof}
Let $G \in {\cal G}_6$ \AY{be arbitrary with vertex-weights $(a^*,d^*)$.}
Let $X_i = \{ v \; | \; d_G(v)=i \}$ for all $i=0,1,2,3$.
We may assume that $G$ is connected as otherwise we can just consider each connected component separately (it will never be advantageous to add edges
between different components). We may clearly also assume that $G$ has at least two vertices, which implies that $X_0 = \emptyset$.
Let $X_2'$ denote all vertices in $X_2$ with an edge to a vertex in $X_3$ and let $X_2'' = X_2 \setminus X_2'$.
Let $Z$ contain all vertices, $z$, in $X_2''$, such that $z$ is adjacent to a vertex, $w$, with $\add{}(w)=0$.  In other words, 
$Z$ contains all vertices from $X_2''$ which are neighbours of a vertex which we are not allowed to add edges to.
Define $Y$ as follows:
 $Y = Z \cup X_2' \cup X_3.$

We will show that $(G,\add{},\del{})$ is a YES-instance to  $(1,1)$-{\sc Cluster Editing}, if and only if there exists a matching in $G$ that covers all vertices in $Y$ but no vertex, $w$,
 with $\del{}(w)=0$.
This will give us the desired polynomial algorithm by Lemma~\ref{thm:EdmondBlossum}.

\2

First assume that $D$ is a matching in $G$ that covers all vertices in $Y$ but no vertex, $w$, with $\del{}(w)=0$.
We may assume that $D$ is maximal, as otherwise we just keep adding edges to $D$ (where both end-points have $d^*$-value 1) until it becomes
maximal (which is not necessarily maximum). Let $G' = G-D$. 
By Lemma~\ref{lem:posII} it suffices to show that every component, $C$, in $G'$ is 
isomorphic to $P_1$ or $P_2$ or $P_3$ and if $C\cong P_3$ then the edge of $\overline{G}$ between the end-vertices of $C$ is addable to $G$.
As every vertex in $X_3$ belongs to $Y$ we note that $\Delta(G') \leq 2$. 

For the sake of contradiction assume that $uv \in E(G')$ and $d_{G'}(u)=d_{G'}(v)=2$. 
As all vertices in $X_2'$ will have degree 1 in $G'$ we note that neither $u$ nor $v$ belong to $X_2'$.
As all vertices in $X_3$ are only adjacent to vertices in $X_2'$, by the definition of ${\cal G}_6$ and $X_2'$, we note that 
neither $u$ nor $v$ belong to $X_3$. So  $u$ and $v$ must both belong
to $X_2''$. By the definition of ${\cal G}_6$ we note that $\del{}(u)=\del{}(v)=1$.
 However, this contradicts the fact that $D$ is maximal, as we could have added the edge $uv$ to $D$.
So, no $uv \in E(G')$ has $d_{G'}(u)=d_{G'}(v)=2$. This implies that all components in $G'$ are isomorphic to $P_1$ or $P_2$ or $P_3$.

Assume that $C$ is a component in $G'$ isomorphic to $P_3$ and let $v_1 v_2 v_3$ be the $3$-path in $C$.
By the construction of $G'$ we note that $v_2 \in X_3$ or $v_2 \in X_2''$. 
If $v_2 \in X_3$ then, by the definition of ${\cal G}_6$, we note that $\add{}(v_1)=\add{}(v_3)=1$.
And if $v_2 \in  X_2''$, then $v_2 \not\in Y$ which by the definition of $Z$ implies that $\add{}(v_1)=\add{}(v_3)=1$.
So in both cases $\add{}(v_1)=\add{}(v_3)=1$, which by Lemma~\ref{lem:posII}, implies that $(G,\add{},\del{})$ is a YES-instance to  $(1,1)$-{\sc Cluster Editing},
as desired.

\2
 
Now conversely assume that $(G,\add{},\del{})$ is a YES-instance to $(1,1)$-{\sc Cluster Editing} and that $D$ denotes the edges deleted from $G$ in the solution. 
By definition $D$ is a matching in $G$ and by Lemma~\ref{lem:posII} every component, $C$, in $G-D$ is
isomorphic to $P_1$ or $P_2$ or $P_3$ and if $C\cong P_3$ then the edge of $\overline{G}$ between the end-vertices of $C$ is addable to $G$.
Let $G^* = G-D$.
As $\Delta(G^*) \leq 2$ we note that $D$ covers all vertices in $X_3$.
Let $u \in X_2'$ be arbitrary and let $v$ be a neighbour of $u$ in $G$ such that $v \in X_3$ (which exists by the definition of $X_2'$). 
Let $N(u)=\{x,v\}$ and let $N(v)=\{u,s,t\}$, where as $G$ is $\{C_3,C_4\}$-free implies that $x,v,u,s,t$ are all distinct.
If no edge incident with $u$ belongs to $D$, then either $xuvs$ or $xuvt$ will be a $P_4$ in $G^*$, a contradiction, so $u$ is covered by $D$.
So all vertices in $X_2'$ are covered by $D$ and previously we showed that all vertices in $X_3$ are also covered by $D$.

Let $z \in Z$ be arbitrary. That is $z \in X_2''$ and $z$ is adjacent to a vertex, $w$, with $\add{}(w)=0$.
Let $N(z)=\{w,q\}$ and note that if $z$ is not covered by $D$ then $wzq$ is a $P_2$ in $G^*$, but as $\add{}(w)=0$, the edge $wq \in \overline{G}$
 is not addable to $G$. So $z$ must be covered by $D$. This implies that $D$ is a matching covering $Y$ in $G$.
Furthermore, by definition, $D$ does not cover any vertex with $d$-value zero.

\2

We have therefore shown that 
$(G,\add{},\del{})$ is a YES-instance to  $(1,1)$-{\sc Cluster Editing}, if and only if there exists a matching in $G$ that covers all vertices in $Y$ but no vertex, $w$, with $\del{}(w)=0$.
\end{proof}

The following theorem is the main result of the paper.

\begin{theorem}\label{cor:polAlg}
 $(1,1)$-{\sc Cluster Editing} can be solved in polynomial time.
\end{theorem}
\begin{proof}
This follows immediately  from Lemmas~\ref{lem:tf}, \ref{lem:noC4with4deg3}, \ref{lem:noC4with3deg3}, \ref{lem:noC4}, \ref{lem:dZero}
 and \ref{lem:polAlg}.
\end{proof}

\section{Series of reductions to \AZ{special} $\{C_3,C_4\}$-free graphs of maximum degree at most 3}\label{sec:5}

In this section we give proofs of lemmas stated in Section \ref{sec:scheme} and additional lemmas used in the proofs of lemmas in Section \ref{sec:scheme}. 

Recall that ${\cal G}_1$ denotes the set of all graphs and 
${\cal G}_2$ denotes the set of all $C_3$-free graphs of maximum degree at most 3.

\vspace{2mm}

\noindent {\bf Lemma \ref{lem:tf}.}
{\em  $(1,1)$-{\sc Cluster Editing} can be reduced from ${\cal G}_1$ to ${\cal G}_2$.}
\begin{proof}
Let $(G,\add{},\del{})$ be an instance of  $(1,1)$-{\sc Cluster Editing}. Suppose that $G$ has a triangle $T\cong C_3$. We can delete at most one edge from $T$. Thus, if $(G,\add{},\del{})$ has a solution, it must contain $T$ as (part of) a clique. In the algorithm below to shorten its description, once a deletable edge $uv$ is deleted from $G$, we immediately set $\del{}(u)=\del{}(v)=0$ and if an addable edge is added to $G$, we  immediately set $\add{}(u)=\add{}(v)=0.$ 

Let $Q:=T$ and execute the following loop \AZ{which may add vertices to $Q$}. While there is a vertex outside $Q$ adjacent to a vertex in $Q$, 
for every such vertex $v$ do the following cases in turn.
If there is only one edge $e$ between $v$ and $Q$ then if $e$ is an deletable edge, then $e$ is deleted and the loop is continued.
 If there are exactly $|V(Q)|-1$ edges between $v$ and $V(Q)$ in $G$ and the edge in $\overline{G}$ between $v$ and $V(Q)$ is addable, then   add the addable edge to $G$ and continue the loop.  
If there are $|V(Q)|$ edges between $v$ and $V(Q)$ in $G$, then add vertex $v$ and the edges between $v$ and $Q$  to $Q$.
 Otherwise, $v$ has between two and $|V(Q)|-2$ edges to $Q$ and we conclude that $(G,\add{},\del{})$ is a NO-instance and stop the loop. 
If the loop stops without concluding that $(G,\add{},\del{})$ is a NO-instance, then delete $Q$ from $G.$
 
 We can continue as above and either eliminate all triangles from $G$ or conclude that $(G,\add{},\del{})$ is a NO-instance. In the former case, we obtain a $C_3$-free instance $(G,\add{},\del{})$ which is a YES-instance if and only if the initial instance is. Note that it takes a polynomial time to compute the $C_3$-free instance or conclude that the initial instance is a NO-instance.
 
 Now suppose that $(G,\add{},\del{})$ is a $C_3$-free instance and $G$ has a vertex $v$ of degree at least four. Since we can delete at most one edge incident to $v$ and $G$ is $C_3$-free, we cannot make a clique including $v$ without adding at least two edges incident to a neighbour of $v$. Thus, $(G,\add{},\del{})$ is a NO-instance. 
 It is not hard to verify that the algorithm runs in polynomial time.
 \end{proof}

Let $(G,\add{},\del{})$ be an instance of  $(1,1)$-{\sc Cluster Editing} and let $Q \subseteq V(G)$ be arbitrary.
Consider the graph $G-Q$, where we delete $Q$ (and all edges incident with $Q$).
We define $\del{Q}$ such that $\del{Q}(v) = \del{}(v) - |N(v) \cap Q|$ for all $v \in V(G-Q)$.
We define $\add{Q}$ to be the function $\add{}$ restricted to $V(G-Q)$. 
\AY{Note that by definition of vertex-weights,
if $\del{Q}(v) <0$ for any $v \in V(G-Q)$ then $(G-Q,\add{Q},\del{Q})$ is a NO-instance.}
\AZ{Similarly, for $F\subseteq E(G)$, let}
$\del{F}(v) = \del{}(v) - |E(v) \cap F|$, where $E(v)$ is the set of edges incident to $v$. Also, $\add{F}=\add{}.$ 
Again, if $\del{F}(v) <0$ for any $v \in V(G)$ then $(G-F,\add{F},\del{F})$ is a NO-instance.

We say that a vertex set $X$ of $G$ be {\em R-deletable} if $X$ has constant size and every solution to $(G,a^*,d^*)$ (if any) puts all edges between $X$ and $V(G)-X$ in $D.$
Note that vacuously this implies that if $G$ is a NO-instance then any set  $X$ of constant size is an R-deletable set. 

\vspace{2mm}

\AZ{The following lemma is used in the reductions with $X$ containing at most 12 vertices.}

\begin{lemma}\label{lem:Rdel}
Let $(G,\add{},\del{})$ be an instance of  $(1,1)$-{\sc Cluster Editing}.  
If a vertex set $X$ of $G$ \AY{is R-deletable}, then we can either solve $(G,\add{},\del{})$ is constant time or reduce the instance to $G-X$. 
\end{lemma}
\begin{proof}
 If $X=V(G)$ then we can solve the problem in constant time, so assume that $X \not= V(G).$ \AY{Let $F$ be all edges between $X$ and $V(D)-X$ and let $G' =G-F.$ 
 If any $\del{F}$ value}  drops below 0 we have a NO-instance by definition of vertex-weights. 
Otherwise, $G$ is a YES-instance if and only if $G'$ is a YES-instance. 
And as $X$ has constant size we can determine if there is a solution to $G'[X]$ in constant time. 
If there is no solution to $G'[X]$ then \AY{$G$ is a NO-instance} and if there is a solution to $G'[X]$, then we can reduce $G$ by deleting $X$, as $G$ is a YES-instance if and only if $G'-X=G-X$ is a YES-instance.

\AY{Note that the above also holds when $(G,\add{},\del{})$ is a NO-instance as then either $G[X]$ is a NO-instance (after adjustment of $d^*$), which can be decided in constant time, 
or  $G-X$ is a NO-instance (after adjustment of $d^*$).}
\end{proof}

 \begin{lemma}\label{lem:posI}
 Let $G$ be a $C_3$-free graph with maximum degree at most 3. 
 Then $(G,\add{},\del{})$ is a YES-instance of $(1,1)$-{\sc Cluster Editing} if and only if there is  a matching $D$ of deletable edges from $G$ such that every connected component $C$ 
of $G-D$ is isomorphic to $P_1$ or $P_2$ or $P_3$ or $C_4$ and if $C\cong P_3$ then the edge of $\overline{G}$ between the end-vertices of $C$ is 
addable to $G$ and if $C\cong C_4$ then the two chords \AY{of $C$, which lie in $\overline{G}$,} are addable to $G$. 
 \end{lemma}
 \begin{proof}
Let $(G,\add{},\del{})$ be a YES-instance of $(1,1)$-{\sc Cluster Editing} and $G-D+A$ is a \AZ{corresponding} solution.
 Since $G$ is a $C_3$-free graph with maximum degree at most 3, each connected component of $G-D+A$ is isomorphic to $K_1$ or $K_2$ or $K_3$ or $K_4$.
 Now the lemma follows from the fact that $G$ is a $C_3$-free graph.
 \end{proof}

\AZ{
\begin{lemma}\label{lem:K23}
Let $G$ be a graph of maximum degree at most 3  containing a $K_{2,3}$ as a subgraph.
Then $(G,\add{},\del{})$ is a NO-instance of $(1,1)$-{\sc Cluster Editing} for every $(a^*,d^*)$.
\end{lemma}
\begin{proof}
Let $G$ contain an induced subgraph $H$ with five vertices containing $K_{2,3}$. Since the maximum degree of $G$ is at most 3, $H$ is either isomorphic to $K_{2,3}$ or $K_{2,3}+e$, where
$e$ is an edge between two vertices of the partite set of $K_{2,3}$ of size 3. Note that $H$ has a vertex $v$ of degree 2.  
The vertices of $H$ cannot be in two or more cliques of a solution as we can only delete a matching $M$ from $H$ and $H-M$ is connected. 
All vertices of $H$ cannot be in a clique of a solution as this would require adding at least two edges incident to $v$, which is not allowed.
\end{proof}
}

 

Recall that ${\cal G}_3$ denotes all $C_3$-free graphs of maximum degree at most 3 which contain no $4$-cycle, all of whose vertices have degree 3.

\vspace{2mm}
 
\noindent{\bf Lemma \ref{lem:noC4with4deg3}.}
{\em  $(1,1)$-{\sc Cluster Editing} can be reduced from ${\cal G}_2$ to ${\cal G}_3$ in polynomial time.}
\begin{proof}
Let $G \in {\cal G}_2$ \AY{be arbitrary with vertex-weights $(a^*,d^*)$.} If there is no $4$-cycle in $G$ where all vertices have degree three then we are done, so
let $C = v_1 v_2 v_3 v_4 v_1$ be a $4$-cycle in $G$ with $d(v_i)=3$ for all $i=1,2,3,4$.
Let $w_i$ be the neighbour of $v_i$ in $G$ which does not lie on $C$. 
Let $V_8=\{v_1,v_2,v_3,v_4,w_1,w_2,w_3,w_4\}$.

\AY{We will show that  $|V_8|=8$ (or we have a NO-instance), so assume for the sake of contradiction that this is not the case. As $G$ is $C_3$-free, we must have $w_1=w_3$ or $w_2=w_4$. 
In either case $K_{2,3}$ is a subgraph of $G$ and we have a NO-instance by Lemma~\ref{lem:K23}. Therefore we may assume that $|V_8|=8$.}

We now prove the following claims.

\2

{\bf Claim A:} {\em If $(G,\add{},\del{})$ is a YES-instance and  $G-D+A$ is a solution, then at least one of the following options holds.

\begin{description}
\item[(i):] $v_i w_i \in D$ for all $i=1,2,3,4$ and no edge on $C$ belongs to $D$.
\item[(ii):] $v_i w_i \not\in D$ for all $i=1,2,3,4$ and $v_1 v_2, v_3 v_4 \in D$ and $v_2 v_3, v_4 v_1 \not\in D$.
\item[(iii):] $v_i w_i \not\in D$ for all $i=1,2,3,4$ and $v_2 v_3, v_4 v_1 \in D$ and $v_1 v_2, v_3 v_4 \not\in D$.
\end{description}

Furthermore, for every such $4$-cycle in $G$ either zero or two edges of the cycle belong to $D$.}

\2

{\bf Proof of Claim A:} Let $C'$ be any $4$-cycle in $G$. If one edge $e \in E(C')$ belongs to $D$, then at least two edges of $C'$ must belong to $D$, as
otherwise $C' - e$ is a $P_4$ in $G-D$, which is not part of a $4$-cycle in $G-D$, a contradiction by Lemma \ref{lem:posI}. As each vertex is adjacent to at most one edge from $D$,
we note that we cannot have more than two edges from $C'$ in $D$.

Option (i) now corresponds to the case when no edge from $C$ belongs to $D$ and Options (ii) and (iii) correspond to the two possible ways that we can add two
(non-adjacent) edges of $C$ to $D$. This completes the proof of Claim~A.\smallqed{}

\2

{\bf Claim B:} {\em We may assume that $\del{}(v_i)=\add{}(v_i)=1$ for all $i=1,2,3,4,$ as otherwise we can reduce our instance or solve it in polynomial time.} 

\2 

{\bf Proof of Claim B:} If $\del{}(v_i)=0$ for some $i\in [4]$, then none of the options in Claim~A is possible, so we have a NO-instance. 
So we may assume that  $\del{}(v_i)=1$ for all $i\in [4]$.

For the sake of contradiction, assume that $\add{}(v_1)=0$ and $(G,\add{},\del{})$ is a YES-instance with a solution $G-D+A$. In this case Option~(i) of Claim A cannot hold as we cannot add the edge $v_1 v_3$ to $C$. 
Option~(ii) of Claim A cannot hold as we cannot add the edge \AZ{$v_1 w_4$} to $G-D$.
Also, Option~(iii) of Claim A cannot hold as we cannot add the edge \AZ{$v_1 w_2$} to $G-D$. So in this case we have a NO-instance. Analogously, we may assume
that $\add{}(v_i)=1$ for all $i\in [4]$, which completes the proof of Claim~B.\smallqed{}

\2

Let $e_1 = w_1 w_2$, $e_2 = w_2 w_3$, $e_3 = w_3 w_4$ and $e_4 = w_4 w_1$ be edges that may or may not exist in $G$. We now prove the following claims.

\2

{\bf Claim C:} \AZ{ {\em We have the following:}

 $\bullet$ {\em We may assume that $e_1 \not\in E(G)$,  as otherwise 
we can reduce or solve our instance in polynomial time. }

 $\bullet$ {\em Option~(iii) in Claim~A cannot occur in any solution of $(G,a^*,d^*)$.}

 $\bullet$ {\em If we know that {there is no solution of $(G,a^*,d^*)$ for which Option (ii) of Claim A occurs,} 
then we can reduce or solve our instance in polynomial time. }

 $\bullet$ {\em Also,  if we know that {there is no solution of $(G,a^*,d^*)$ for which Option (i) of Claim A occurs}, 
then we can reduce or solve our instance in polynomial time. 
}
}

\2

{\bf Proof of Claim C:}  We may assume that $G$ is connected; otherwise consider the connected component containing $V_8.$ 
By renaming vertices we may assume that $e_1 \not\in E(G)$, unless $e_i \in E(G)$ for all $i \in [4]$. 
So assume that $e_i \in E(G)$ for all $i \in [4]$, which, as $G$ is connected and of maximum degree three, implies that $V(G)=V_8$.
In this case we can determine if we have a YES-instance (which is the case if and only if $\del{}(v)=\add{}(v)=1$ for all $v \in V(G)$) or NO-instance in constant time, 
and so we would be done. So we may indeed assume that $e_1 \not\in E(G)$. 

If Option~(iii) in Claim~A occurs, then $w_1 v_1 v_2 w_2$ is a $P_4$ in $G-D$ where $w_1w_2 \not\in E(G)$, which is not possible by \GG{Lemma \ref{lem:posI}}, so Option~(iii) in Claim~A cannot occur.

\AZ{We may assume that, in the rest of the proof, $(G,a^*,d^*)$ is a YES-instance as otherwise the remaining parts of Claim C vacuously hold.}
\AY{As we may assume that Option~(iii) in Claim~A cannot occur, either Option (i) or Option (ii) of Claim A must occur for every solution of $(G,\add{},\del{})$.
First assume that we can show that Option~(ii) does not occur, and therefore Option~(i) must occur in all solutions. In this case, $w_iv_i \in D$ for all $i \in [4]$ and
therefore $V(C)$ is a R-deletable set, so we can reduce $G$ to $G-V(C)$ by Lemma~\ref{lem:Rdel}.

Now assume that Option~(i) cannot occur which implies that Option~(ii) must occur in all solutions. 
This implies that all edges between $V_8$ and $V(G) \setminus V_8$ must belong to $D$. Thus, either $V(G)=V_8$ in which case we are done or
$V_8$ is an R-deletable set, so we can reduce $G$ to $G-V_8$ by Lemma~\ref{lem:Rdel}.
This completes the proof of Claim~C.\smallqed{} }

\2

\begin{figure}[hbtp]
\begin{center}
\tikzstyle{vertexB}=[circle,draw, minimum size=14pt, scale=0.6, inner sep=0.5pt]
\tikzstyle{vertexR}=[circle,draw, color=red!100, minimum size=14pt, scale=0.6, inner sep=0.5pt]

\begin{tikzpicture}[scale=0.3]
 \node (v1) at (3,7) [vertexB] {$v_1$};
 \node (v2) at (9,7) [vertexB] {$v_2$};
 \node (v3) at (9,3) [vertexB] {$v_3$};
 \node (v4) at (3,3) [vertexB] {$v_4$};

 \node (w1) at (0,9) [vertexB] {$w_1$};
 \node (w2) at (12,9) [vertexB] {$w_2$};
 \node (w3) at (12,1) [vertexB] {$w_3$};
 \node (w4) at (0,1) [vertexB] {$w_4$};

\draw (6,9.9) node {{\tiny No edge}};
\draw (6,0.1) node {{\tiny No edge}};

\draw (-1,5) node {{\tiny $e_4$}};
\draw (13,5) node {{\tiny $e_2$}};


\draw [line width=0.03cm] (v1) to (v2);
\draw [line width=0.03cm] (v2) to (v3);
\draw [line width=0.03cm] (v3) to (v4);
\draw [line width=0.03cm] (v4) to (v1);

\draw [line width=0.03cm] (w1) to (v1);
\draw [line width=0.03cm] (w2) to (v2);
\draw [line width=0.03cm] (w3) to (v3);
\draw [line width=0.03cm] (w4) to (v4);

\draw [line width=0.03cm] (w1) to (w4);
\draw [line width=0.03cm] (w2) to (w3);

\draw [dotted, line width=0.03cm] (w1) to (w2);
\draw [dotted, line width=0.03cm] (w3) to (w4);
\draw (6,-2) node {(a)};
\end{tikzpicture} \hspace{2cm}
\begin{tikzpicture}[scale=0.3]
 \node (v1) at (3,7) [vertexB] {$v_1$};
 \node (v2) at (9,7) [vertexB] {$v_2$};
 \node (v3) at (9,3) [vertexB] {$v_3$};
 \node (v4) at (3,3) [vertexB] {$v_4$};
 \node (w1) at (0,9) [vertexB] {$w_1$};
 \node (w2) at (12,9) [vertexB] {$w_2$};
 \node (w3) at (12,1) [vertexB] {$w_3$};
 \node (w4) at (0,1) [vertexB] {$w_4$};
 \node (z1) at (-6,9) [vertexB] {$z_1$};
 \node (z4) at (-6,1) [vertexB] {$z_4$};

\draw (6,9.9) node {{\tiny No edge}};
\draw (6,0.1) node {{\tiny No edge}};
\draw [line width=0.03cm] (v1) to (v2);
\draw [line width=0.03cm] (v2) to (v3);
\draw [line width=0.03cm] (v3) to (v4);
\draw [line width=0.03cm] (v4) to (v1);

\draw [line width=0.03cm] (w1) to (v1);
\draw [line width=0.03cm] (w2) to (v2);
\draw [line width=0.03cm] (w3) to (v3);
\draw [line width=0.03cm] (w4) to (v4);

\draw [line width=0.03cm] (w1) to (w4);
\draw [line width=0.03cm] (w2) to (w3);

\draw [dotted, line width=0.03cm] (w1) to (w2);
\draw [dotted, line width=0.03cm] (w3) to (w4);
\draw [dotted, line width=0.03cm] (w3) to (w4);

\draw [line width=0.03cm] (z1) to (w1);
\draw [line width=0.03cm] (z4) to (w4);
\draw (4,-2) node {(b)};
\end{tikzpicture}

\caption{Graph (a) is an illustration of Claim~D in the proof of Lemma~\ref{lem:noC4with4deg3}, while Graph~(b) is an illustration of Claim~F.}\label{fig:pic1}
\end{center}
\end{figure}
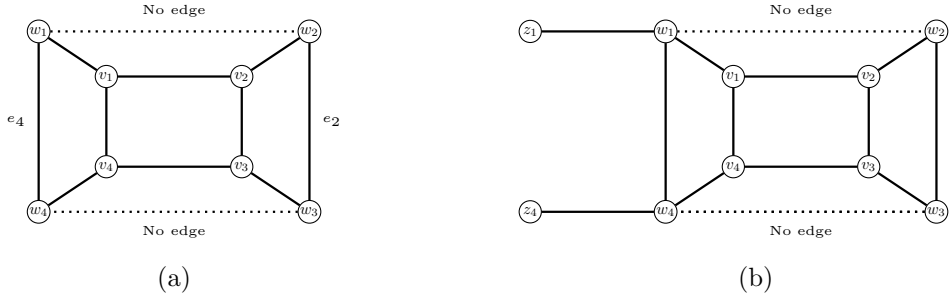

{\bf Claim D:} {\em We may assume that $e_2,e_4 \in E(G)$ and $e_1,e_3 \not\in E(G),$  as otherwise we can reduce our instance or solve it in polynomial time. 
(See Figure~\ref{fig:pic1}(a) for an illustration.)
We may also assume that $\del{}(w_i)=\add{}(w_i)=1$ for all $i \in \{1,2,3,4\}$.}

\2

{\bf Proof of Claim D:} By Claim~C we may assume that $e_1 \not\in E(G)$. For the sake of contradiction assume that $e_2 \not\in E(G)$.
By Claim~C we note that part~(iii) in Claim~A cannot occur.
Analogously, if Option~(ii) in Claim~A occurs, then $w_2 v_2 v_3 w_3$ is a $P_4$ in $G-D$ where $e_2=w_2w_3 \not\in E(G)$, which is not possible.
So, Option~(ii) in Claim~A can also not occur \AY{and we are done by Claim~C}.
Hence, we may assume that $e_2 \in E(G)$. The fact that $e_4 \in E(G)$ can be proved analogously, so we may assume that $e_2,e_4 \in E(G)$.

For the sake of contradiction assume that $e_3 \in E(G)$. In this case Option~(i) in Claim~A cannot occur, as $w_2 w_3 w_4 w_1$ would then be a $P_4$ in 
$G-D$ where $w_1 w_2 \not\in E(G-D)$, a contradiction. \AY{So, again we are done by Claim~C.}



So we may assume that $e_3 \not\in E(G)$, as otherwise we either reduce our instance or solve it in constant time.
This completes the proof of the first part of Claim~D. In order to prove the second part we consider the following two cases.

First assume that $\del{}(w_i)=0$ for some $i \in [4]$. Then Option~(i) in Claim~A cannot occur, so by Claim~C, Option~(ii) must occur.
\AZ{Thus, we can reduce $G$ by Claim C.}
Hence, we may assume that $\del{}(w_i)=1$ for all $i \in [4]$.

Now assume that $\add{}(w_i)=0$ for some $i \in [4]$. In this case Option~(ii) in Claim~A cannot occur, so by Claim~C, Option~(i) must occur.
\AZ{Thus, we can reduce $G$ by Claim C.}
So we may assume that $\add{}(w_i)=1$ for all $i \in [4]$.
This completes the proof of Claim~D.\smallqed{}

\2

{\bf Definition:} If $d(w_1)=3$ then let $z_1$ be defined such that $N(w_1)=\{z_1,v_1,w_4\}$. 

 If $d(w_2)=3$ then let $z_2$ be defined such that $N(w_2)=\{z_2,v_2,w_3\}$.

 If $d(w_3)=3$ then let $z_3$ be defined such that $N(w_3)=\{z_3,v_3,w_2\}$.

 If $d(w_4)=3$ then let $z_4$ be defined such that $N(w_4)=\{z_4,v_4,w_1\}$.

Note that the $z_i$'s may not be distinct and we may have $z_i=w_j$ for some $i$ and $j$. But, by \AY{Claim~D we have $z_1 \not= w_2$, $z_2 \not= w_1$,
$z_3 \not= w_4$ and $z_4 \not= w_3$.}

\2

{\bf Claim E:} {\em We may assume that if $z_i$ exists then $z_i \not\in V_8$ for all $i=1,2,3,4$ and if $z_i$ and $z_j$ both exist with $i \not= j$, then $z_i \not= z_j$. Otherwise, we can reduce our instance or solve it in polynomial time.}

\2

{\bf Proof of Claim E:} For the sake of contradiction assume that $z_1$ exists and $z_1 \in V_8$. By definition there exists a $j \in [4]$ such that $z_i = w_j$.
By Claim~D, we must have $j=3$ (and therefore $z_1=w_3$ and $z_3=w_1$). First consider the case when $z_2$ exists and $z_2=w_4$. In this case, as $G$ is connected, 
we have $V(G)=V_8$ and we can determine if $(G, a^*,d^*)$ is a YES-instance or a NO-instance in constant time. So we may assume that $w_2w_4 \not\in E(G)$. 
Option~(i) in Claim~A cannot occur, as $w_4 w_1 w_3 w_2$ would then be a $P_4$ in
$G-D$ where $w_2 w_4 \not\in E(G-D)$, a contradiction. So, by Claim~C, \AY{we are done in this case.} 
Analogously we can show that \AY{$z_i \not\in V_8$ for all $i=1,2,3,4$ (or we can reduce our instance or solve it in polynomial time).

Now for the sake of contradiction assume that $z_i$ and $z_j$ both exist and $i \not= j$ but $z_i=z_j$. 
Option~(ii) in Claim~A cannot occur, as in this case
$w_i z_i \in D$ and $w_j z_j \in D$, but $z_i$ can be incident with at most one edge from $D$. So, we are done by Claim~C.}
This completes the proof of Claim~E.\smallqed{}

\begin{figure}[hbtp]
\begin{center}
\tikzstyle{vertexB}=[circle,draw, minimum size=14pt, scale=0.6, inner sep=0.5pt]
\tikzstyle{vertexR}=[circle,draw, color=red!100, minimum size=14pt, scale=0.6, inner sep=0.5pt]

\begin{tikzpicture}[scale=0.3]
 \node (v1) at (3,7) [vertexB] {$v_1$};
 \node (v2) at (9,7) [vertexB] {$v_2$};
 \node (v3) at (9,3) [vertexB] {$v_3$};
 \node (v4) at (3,3) [vertexB] {$v_4$};
 \node (w1) at (0,9) [vertexB] {$w_1$};
 \node (w2) at (12,9) [vertexB] {$w_2$};
 \node (w3) at (12,1) [vertexB] {$w_3$};
 \node (w4) at (0,1) [vertexB] {$w_4$};
 \node (z1) at (-4,10) [vertexB] {$z_1$};
 \node (z3) at (16,0) [vertexB] {$z_3$};
 \node (z4) at (-4,0) [vertexB] {$z_4$};
\draw (6,9.5) node {{\tiny No edge}};
\draw (6,1.5) node {{\tiny No edge}};
\draw (6,0.5) node {{\tiny No edge}};
\draw [line width=0.03cm] (v1) to (v2);
\draw [line width=0.03cm] (v2) to (v3);
\draw [line width=0.03cm] (v3) to (v4);
\draw [line width=0.03cm] (v4) to (v1);
\draw [line width=0.03cm] (w1) to (v1);
\draw [line width=0.03cm] (w2) to (v2);
\draw [line width=0.03cm] (w3) to (v3);
\draw [line width=0.03cm] (w4) to (v4);
\draw [line width=0.03cm] (w1) to (w4);
\draw [line width=0.03cm] (w2) to (w3);
\draw [dotted, line width=0.03cm] (w1) to (w2);
\draw [dotted, line width=0.03cm] (w3) to (w4);
\draw [dotted, line width=0.03cm] (z3) to (z4);

\draw [line width=0.03cm] (z1) to (z4);

\draw [line width=0.03cm] (z1) to (w1);
\draw [line width=0.03cm] (z3) to (w3);
\draw [line width=0.03cm] (z4) to (w4);
\draw (6,-2) node {(a)};
\draw (19,5) node {{\Large $\Rightarrow$}};
\end{tikzpicture}
\begin{tikzpicture}[scale=0.3]
\draw (-7,5) node {\mbox{ }};

 \node (w1) at (0,9) [vertexB] {$w_1$};
 \node (w2) at (12,9) [vertexB] {$w_2$};
 \node (w3) at (12,1) [vertexB] {$w_3$};
 \node (w4) at (0,1) [vertexB] {$w_4$};

 \node (z1) at (-4,10) [vertexB] {$z_1$};
 \node (z3) at (16,0) [vertexB] {$z_3$};
 \node (z4) at (-4,0) [vertexB] {$z_4$};

\draw (6,0.5) node {{\tiny No edge}};

\draw [line width=0.03cm] (w1) to (w2);
\draw [line width=0.03cm] (w2) to (w3);
\draw [line width=0.03cm] (w3) to (w4);
\draw [line width=0.03cm] (w4) to (w1);

\draw [dotted, line width=0.03cm] (z3) to (z4);
\draw [line width=0.03cm] (z1) to (z4);

\draw [line width=0.03cm] (z1) to (w1);
\draw [line width=0.03cm] (z3) to (w3);
\draw [line width=0.03cm] (z4) to (w4);
\draw (6,-2) node {(b)};
\end{tikzpicture}

\caption{The reduction in Claim~G of Lemma~\ref{lem:noC4with4deg3}.}\label{fig:pic2}
\end{center}
\end{figure}
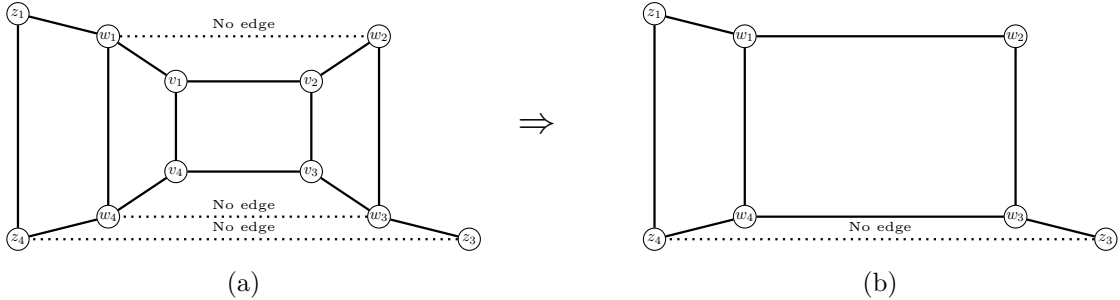

\2

{\bf Claim F:} {\em If $z_1$ and $z_4$ both exist then $z_1 z_4 \in E(G)$. Otherwise, we can reduce our instance or solve it in polynomial time.
Analogously, we may assume that if $z_2$ and $z_3$ both exist then $z_2 z_3 \in E(G)$.}

\2

{\bf Proof of Claim F:} For the sake of contradiction assume that $z_1$ and $z_4$ both exist, but $z_1 z_4 \not\in E(G)$ (see Figure~\ref{fig:pic1}(b)). 
By Claim~C, either Option~(i) or Option~(ii) in Claim~A must occur. 
Assume that Option~(i) in Claim~A occurs. This implies that $v_i w_i \in D$ for all $i \in [4]$.
Hence, $z_1 w_1 w_4 z_4$ is a $P_4$ in $G-D$ where $z_1 z_4 \not\in E(G)$, a contradiction.
So Option~(ii) in Claim~A must occur.
We are now done by \AY{Claim~C}. 
This completes the proof of Claim~F.\smallqed{}


\2

{\bf Claim G:} {\em If $z_1$, $z_3$ and $z_4$ all exist, then  we may assume that $z_3 z_4 z_1$ is a path in $G$. Otherwise, we can reduce our instance or solve it in polynomial time.
Analogously if  $z_1$, $z_2$ and $z_4$ all exist, then we may assume that $z_4 z_1 z_2$ is a path in $G$. 
And if  $z_2$, $z_3$ and $z_4$ all exist, then that $z_2 z_3 z_4$ is a path in $G$.  
And finally if $z_1$, $z_2$ and $z_3$ all exist, then that $z_1 z_2 z_3$ is a path in $G$. }

\2

{\bf Proof of Claim G:} Assume that $z_1$, $z_3$ and $z_4$ exist.
By Claim~F we may assume that $z_1 z_4 \in E(G)$.
For the sake of contradiction assume that $z_3 z_4 \not\in E(G)$.
Let $G^+$ be obtained from $G$ by deleting the vertices $v_1,v_2,v_3,v_4$ and adding the edges $w_1 w_2$ and $w_3 w_4$ (the vertex weights $\del{}$ and $\add{}$ remain unchanged).
See Figure~\ref{fig:pic2} for an illustration. 
We will now show that $G^+$ is a YES-instance if and only if $G$ is a YES-instance.

First assume that $(G,a^*,d^*)$ is a YES-instance\AZ{, i.e., there are $A$ and $D$ such that $G-D+A$ is a solution. We will prove that there are $A^+$ and $D^+$ such that $G^+-D^++A^+$ is a solution.}  
By Claim~C we note that Option~(i) or Option~(ii) in Claim~A must occur. 
If Option~(i) occurs then $v_i w_i \in D$ for all $i \in [4]$ in $G$. \AZ{Let $D^+ := D \cup \{w_1w_2, w_3w_4\} \setminus  \{v_1w_1, v_2w_2 ,v_3w_3 ,v_4w_4\}.$}
This way we get exactly the same components in \AZ{$G-D$ and $G^+-D^+$}, except the $4$-cycle $v_1 v_2 v_3 v_4 v_1$ in \AZ{$G-D$}. \AZ{Thus, $G^+-D^++A^+$ is a solution, where $A^+=A\setminus \{v_1v_3,v_2v_4\}.$}
Now assume that Option~(ii) occurs in Claim~A. In this case, \AZ{let $D^+=D\setminus \{v_1v_2,v_3v_4\}$.}
Now we get exactly the same components in $G-D$ and $G^+-D^+$, except the $4$-cycles $v_1 v_4 w_4 w_1 v_1$ and $v_2 v_3 w_3 w_2 v_2$ in \AZ{$G-D$} have been 
replaced by $w_1 w_2 w_3 w_4 w_1$ in \AZ{$G^+-D^+$}. So \AZ{in this case $G^+$ is a YES-instance, too}.

Conversely assume that $G^+$ is a YES-instance. We will now show that $G$ is a YES-instance. Let $D^+$ be the matching in $G^+$ such that 
all components in $G^+-D^+$ are $P_1$, $P_2$, $P_3$ or $C_4$ (and the $P_3$'s and $C_4$'s can be made into cliques by adding a matching of addable edges).
Let $C^+$ be the $4$-cycle $w_1 w_2 w_3 w_4 w_1$ in $G^+$.
As in the proof of Claim~A we note that $D^+$ contains no edge from $C^+$ or it contains exactly two edges from $C^+$, which are non-adjacent.
Note that $w_2w_3, w_4w_1 \in D^+$ is not possible as the path $z_4 w_4 w_3 z_3$ would be a $P_4$ in $G^+ - D^+$ and $z_3 z_4 \not\in E(G)$.
So either $w_1w_2, w_3w_4 \in D^+$ or no edge from $C^+$ is in $D^+$.

First assume that $w_1w_2, w_3w_4 \in D^+$. Then $z_1 w_1 w_4 z_4 z_1$ is a $C_4$ in $G^+ - D^+$ and $w_2 w_3 z_3$ is a $P_2$ in $G^+ - D^+$ (which may be part of 
a $C_4$ if $z_2$ exists). In this case let $D = D^+ \cup \{v_1 w_1, v_2w_2, v_3w_3, v_4w_4\} \setminus \{w_1w_2, w_3w_4\}$ and note that we obtain 
exactly the same components in $G$ and $G^+$, except the $4$-cycle $v_1 v_2 v_3 v_4 v_1$ in $G$.

So we now consider the case when no edge from $C^+$ is in $D^+$. In this case $w_1 z_1, w_3 z_3, w_4 z_4 \in D^+$. Let  $D = D^+ \cup \{v_1 v_2, v_3 v_4\}$ and note that we obtain 
exactly the same components in $G$ and $G^+$, except the $4$-cycle $w_1 w_2 w_3 w_4 w_1$ in $G^+$ has been replaced by the $4$-cycles $w_1 v_1 v_4 w_4 w_1$ and $w_2 v_2 v_3 w_3 w_2$ in $G$.
We have now proven that $G^+$ is a YES-instance if and only if $G$ is a YES-instance. 
This completes the proof of Claim~G.\smallqed{}

\2

{\bf Claim H:} {\em We may assume that at most one of $z_1$ and $z_4$ exist. Otherwise, we can reduce our instance or solve it in polynomial time. Similarly, at most one of $z_2$ and $z_3$ exist.}

\2

{\bf Proof of Claim H:} Assume that both $z_1$ and $z_4$ exist.
We first consider the case when both $z_2$ and $z_3$ also exist. By Claim~G, we have that $z_1 z_2 z_3 z_4 z_1$ is a $4$-cycle in $G$ and therefore 
 $G$ has 12 vertices, and we can 
determine if it is a YES-instance or a NO-instance in constant time. So we may assume that either none of $z_2$ and $z_3$ exist or exactly one of them exist.

\AY{First} assume that $z_3$ exists, but $z_2$ does not exist. By Claims~\AZ{F and} G, $z_1 z_4, z_3 z_4 \in E(G)$. 
Let $V_{11}=V_8 \cup \{z_1,z_3,z_4\}$. If $V(G)=V_{11}$ then we can solve our instance in constant time so assume that this is not the case. 
We consider the cases when there exists a $w \in V(G) \setminus V_{11}$ with $\{z_1,z_3\} \subseteq N(w)$ and when no such $w$ exists.
If such a $w$ exists then Option~(i) in Claim~A cannot occur, as if it did both the edges $wz_1$ and $wz_3$ would need to be in $D$ \AY{(as $z_1 w_1 w_4 z_4 z_1$ would
be a $C_4$ in $G-D$ and $w_2 w_3 z_3$ would be a $P_3$ in $G-D$),} a contradiction.
\AY{So, we are done by Claim~C.}
We may therefore now assume that no such $w$ exists. As $G$ is connected and $V(G) \not= V_{11}$ we must have a vertex $r \in  V(G) \setminus V_{11}$ 
which is adjacent to either $z_1$ or to $z_3$.  Now Option~(ii) of Claim~A cannot occur, as if it did then the edges $w_1z_1$, $w_3z_3$ and $w_4z_4$ would need to be in $D$,
which implies that either $r z_1 z_4 z_3$ or $z_1 z_4 z_3 r$ is a $P_4$ in $G-D$ where the end-points of the path are not adjacent in $G$ (as no $w$ exists in
$V(G) \setminus V_{11}$ that is adjacent to both $z_1$ and $z_3$).
\AY{So, again we are done by Claim~C.}

This completes all cases and implies that we can reduce when $z_3$ exists, but $z_2$ does not exist.

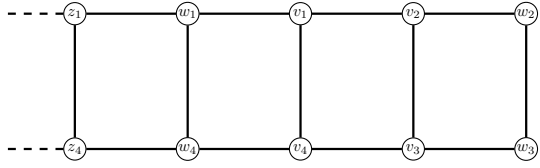
\begin{figure}[hbtp]
\begin{center}
\tikzstyle{vertexB}=[circle,draw, minimum size=14pt, scale=0.6, inner sep=0.5pt]
\tikzstyle{vertexR}=[circle,draw, color=red!100, minimum size=14pt, scale=0.6, inner sep=0.5pt]

\begin{tikzpicture}[scale=0.3]
 \node (v1) at (11,8) [vertexB] {$v_1$};
 \node (v2) at (16,8) [vertexB] {$v_2$};
 \node (v3) at (16,2) [vertexB] {$v_3$};
 \node (v4) at (11,2) [vertexB] {$v_4$};
 \node (w1) at (6,8) [vertexB] {$w_1$};
 \node (w2) at (21,8) [vertexB] {$w_2$};
 \node (w3) at (21,2) [vertexB] {$w_3$};
 \node (w4) at (6,2) [vertexB] {$w_4$};
 \node (z1) at (1,8) [vertexB] {$z_1$};
 \node (z4) at (1,2) [vertexB] {$z_4$};
\draw [line width=0.03cm] (v1) to (v2);
\draw [line width=0.03cm] (v2) to (v3);
\draw [line width=0.03cm] (v3) to (v4);
\draw [line width=0.03cm] (v4) to (v1);
\draw [line width=0.03cm] (w1) to (v1);
\draw [line width=0.03cm] (w2) to (v2);
\draw [line width=0.03cm] (w3) to (v3);
\draw [line width=0.03cm] (w4) to (v4);
\draw [line width=0.03cm] (w1) to (w4);
\draw [line width=0.03cm] (w2) to (w3);

\draw [dashed, line width=0.03cm] (z1) to (-2,8);
\draw [dashed, line width=0.03cm] (z4) to (-2,2);

\draw [line width=0.03cm] (z1) to (z4);

\draw [line width=0.03cm] (z1) to (w1);
\draw [line width=0.03cm] (z4) to (w4);
\end{tikzpicture}

\caption{The graph when $z_2$ and $z_3$ do not exist.}\label{fig:pic3}
\end{center}
\end{figure}

Analogously we can also reduce if $z_2$ exists, but $z_3$ does not exist. So the last case to consider is when neither $z_2$ nor $z_3$ exist.
In this case the graph looks like in Figure~\ref{fig:pic3}. If we had considered the cycle $w_1 v_1 v_4 w_4 w_1$ instead of $v_1 v_2 v_3 v_4 v_1$, 
then we would be in the case where $G$ has 10 vertices (when $d(z_1)=d(z_4)=2$) or we would have been in a case where at least
three $z$-vertices would exist (namely $\{w_2,w_3\} \cup \GG{N(z_1)\cup N(z_4)}\setminus \{z_1,w_1,z_4,w_4\}$), which means we can reduce as seen above.
This completes the proof of Claim~H.\smallqed{}

\2


We now return to the proof of Lemma~\ref{lem:noC4with4deg3}. 
By Claim~H we may assume that at most one of $z_1$ and $z_4$ exist and at most one of $z_2$ and $z_3$ exist.
If no $z_i$ exist then $G$ has 8 vertices, so we can solve our problem in constant time.
Therefore we may, without loss of generality assume that $z_1$ exists. If $z_2$ exists then 
let $z=z_2$ and if $z_3$ exists then let $z=z_3$ and if neither $z_2$ or $z_3$ exist then let $z$ be undefined.
We now create $G^*$ from $G$ by removing the vertices $w_1, w_2, w_3, w_4$ and adding the edge $v_1 z_1$ and if $z$ exists also the edge $v_3 z$.
We will show that $G^*$ is a YES-instance if and only if
$G$ is a YES-instance.

First assume that $G$ is a YES-instance. By Claim~C we note that Option~(i) or Option~(ii) in Claim~A must occur.
If Option~(i) occurs then $v_i w_i \in D$ for all $i \in [4]$ in $G$. Let $D^* = D \cup \{v_1 v_2, v_3 v_4\} \setminus \{v_1w_1, v_2w_2, v_3w_3, v_4w_4 \}$.
This way we get exactly the same components in $G$ and $G^*$, except the $4$-cycle $v_1 v_2 v_3 v_4 v_1$, the $P_3=z_1w_1w_4$ and the 
$P_2$ or $P_3$ on the vertex set $\{w_2, w_3\}$ or $\{w_2, w_3, z\}$ in $G$ get replaced by the $P_3=z_1 v_1 v_4$ and  
either the $P_2=v_2 v_3$ (if $z$ does not exist) or the $P_3=v_2 v_3 z$ in $G^*$.
Now assume that Option~(ii) occurs in Claim~A. In this case  $z_1 w_1 \in D$ and if $z$ exists then $w_2 z$ or $w_3z$ is in $D$.
Also, we let \AZ{$D^* = D \setminus \{v_1 v_2, v_3 v_4\}$}, \AY{where we also replace $z_1w_1$ by $z_1v_1$ and if $z$ exists then the edge $w_2 z$ or $w_3z$ gets
replaced by $v_3z$.}
Now we get exactly the same components in $G-D$ and $G^* - D^*$, except the $4$-cycles $v_1 v_4 w_4 w_1 v_1$ and $v_2 v_3 w_3 w_2 v_2$ in $G$ have been
replaced by $v_1 v_2 v_3 v_4 v_1$ in $G^*$. So in all cases $G^*$ is a YES-instance.

Conversely assume that $G^*$ is a YES-instance. We will now show that $G$ is a YES-instance. Let $D^*$ be the matching in $G^*$ such that
all components in $G^* - D^*$ are $P_1$, $P_2$, $P_3$ or $C_4$ (and the $P_3$'s and $C_4$'s can be made into cliques by adding a matching of addable edges).
As in the proof of Claim~A we note that $D^*$ contains no edge from $C$ or it contains exactly two edges from $C$, which are non-adjacent.

First assume that $D^*$ contains no edge from $C$. Then $z_1 v_1 \in D^*$ and if $z$ exists then $z v_3 \in D$. 
In this case, let $D$ be obtained from \AZ{$D^* \cup \{v_1v_2, v_3v_4,z_1 w_1\} \setminus \{z_1 v_1\}$} by replacing $v_3 z$ by the edge from $\{w_2,w_3\}$ to $z$ if $z$ exists.
Again, we obtain exactly the same components in $G-D$ and $G^* - D^*$, except the $4$-cycle $v_1 v_2 v_3 v_4 v_1$ in $G^*$ has been replaced by the
two $4$-cycles  $v_1 v_4 w_4 w_1 v_1$ and $v_2 v_3 w_3 w_2 v_2$ in $G$.

So we now consider the case when $D^*$ contains exactly two edges from $C$. In this case $G^* - D^*$ contains a $P_3$ containing the edge $z_1 v_1$ and either the edge $v_1 v_2$ or $v_1 v_4$.
Similarly, if $z$ exists, then $G^* - D^*$ contains a $P_3$ containing the edge $z v_3$ and either the edge $v_3 v_4$ or $v_3 v_2$.
In this case let  $D$ be obtained from $D^*$ by deleting these two edges on $C$ that belong to $D^*$ and adding the edges $\{v_1w_1, v_2w_2, v_3w_3, v_4w_4 \}$ instead.
This way we obtain exactly the same components in $G-D$ and $G^* - D^*$, except for the following:
The $P_3$ in $G^* - D^*$ containing $v_1$ now becomes the $P_2$  $z_1 w_1 w_4$ and the $P_3$ containing $z$ (if it exists) now becomes a $P_3$ containing the edge $w_2 w_3$ as well as the 
vertex $z$. Furthermore $G-D$ will contain the $4$-cycle $C$, which is not in $G^* - D^*$.  
This proves that $G^*$ is a YES-instance if and only if $G$ is a YES-instance.
\end{proof}

Recall that ${\cal G}_4$ denotes all $C_3$-free graphs of maximum degree at most 3 which contain no $4$-cycle, with at least three vertices of degree three.

\vspace{2mm}
 
\noindent {\bf Lemma \ref{lem:noC4with3deg3}.}
{\em  $(1,1)$-{\sc Cluster Editing} can be reduced from ${\cal G}_3$ to ${\cal G}_4$ in polynomial time.}
\begin{proof}
Let $G \in {\cal G}_3$ \AY{be arbitrary with vertex-weights $(a^*,d^*)$.}
If there is no $4$-cycle in $G$ where three of the vertices have degree three then we are done, so
let $C = v_1 v_2 v_3 v_4 v_1$ be a $4$-cycle in $G$ with $d(v_i)=3$ for all $i=1,2,3$ and $d(v_4)=2$.
Let $w_i$ be the neighbour of $v_i$ in $G$ which does not lie on $C$ for $i=1,2,3$.
Let $V_7=\{v_1,v_2,v_3,v_4,w_1,w_2,w_3\}$.
\AY{We will show that  $|V_7|=7$ (or we have a NO-instance), so assume for the sake of contradiction that this is not the case. As $G$ is $C_3$-free, we must have $w_1=w_3$,
which implies that $K_{2,3}$ is a subgraph of $G$ and we have a NO-instance by Lemma~\ref{lem:K23}. Therefore we may assume that $|V_7|=7$.}


\GG{The following claim can be proved analogously to Claim~A in Lemma~\ref{lem:noC4with4deg3} and thus its proof is omitted.}

\2

{\bf Claim A:} {\em If $(G,\add{},\del{})$ is a YES-instance and  $G-D+A$ is a solution.
Then one of the following options hold.

\begin{description}
\item[(i):] $v_i w_i \in D$ for all $i=1,2,3$ and no edge on $C$ belongs to $D$.
\item[(ii):] $v_i w_i \not\in D$ for all $i=1,2,3$ and $v_1 v_2, v_3 v_4 \in D$ and $v_2 v_3, v_4 v_1 \not\in D$.
\item[(iii):] $v_i w_i \not\in D$ for all $i=1,2,3$ and $v_2 v_3, v_4 v_1 \in D$ and $v_1 v_2, v_3 v_4 \not\in D$.
\end{description}

Furthermore, for every such $4$-cycle in $G$ either zero or two edges of the cycle belong to $D$.}



\2

\2





{\bf Claim B:} {\em If Option~(ii) in Claim~A occurs, then we must have \AY{$w_2 w_3 \in E(G)$.}
If Option~(iii) occurs in Claim~A, then we must have \AY{$w_1 w_2 \in E(G)$.} }

\2

{\bf Proof of Claim B:} \AY{For the sake of contradiction, assume that Option~(ii) occurs in Claim~A, but $w_2 w_3 \not\in E(G)$.
In this case $w_2 v_2 v_3 w_3$ is a $P_4$ in $G-D$ where $w_2 w_3 \not\in E(G)$, a contradiction.
The fact that if Option~(iii) occurs in Claim~A, then we must have $w_1 w_2 \in E(G)$ is proved analogously.}~\smallqed{}

\2

{\bf Claim C:} {\em We may assume that exactly one of the edges $w_1w_2$ and $w_2w_3$ belongs to $G$.}

\2

{\bf Proof of Claim C:} If $w_1w_2 \not\in E(G)$ and $w_2w_3 \not\in E(G)$ then, by Claim~C, neither Option~(ii) nor Option~(iii) occurs in Claim~A.
Therefore Option~(i) must occur in Claim~A. 
 This implies that $V(C)$ is an R-deletable set, so we can reduce $G$ to $G-V(C)$ by Lemma~\ref{lem:Rdel}.
We may therefore assume that at least one of $w_1w_2$ and $w_2w_3$ belong to $G$.

Assume that $w_1w_2 \in E(G)$ and $w_2w_3 \in E(G)$. As $G \in {\cal G}_3$ we note that the two $4$-cycles $w_1 v_1 v_2 w_2 w_1$ and $w_2 v_2 v_3 w_3 w_2$
both must contain a degree-$2$ vertex. This implies that $d(w_1)=d(w_3)=2$ and $V(G)=V_7$. We can therefore solve the problem in 
constant time. We may therefore assume that one of the edges $w_1w_2$ and $w_2w_3$ do not belong to $G$.
This proves Claim~C.\smallqed{}

\2

By Claim~C, we may without loss of generality assume that  $w_1w_2 \in E(G)$ and $w_2w_3 \not\in E(G)$.

\2

\AY{
{\bf Claim D:} {\em Option~(ii) in Claim~A cannot occur. Furthermore,  we may assume that $d(w_1) \leq 2$ or $d(w_2) \leq 2$. We may also assume that
$d(w_3) \leq 2$ as otherwise we can either solve our problem in polynomial time or reduce the instance. }
}

\2

{\bf Proof of Claim D:} 
Since $w_2w_3 \not\in E(G)$ by Claim B, \AY{Option~(ii)} in Claim~A cannot occur. Note that $d(w_1)=d(w_2)=3$ is not possible as $G \in {\cal G}_3$ and $d(v_1)=d(v_2)=3$.

For the sake of contradiction assume that $d(w_3)=3$ and that $N(w_3)=\{v_3,a,b\}$. Recall that \AY{Option~(ii)}  in Claim~A cannot occur. Assume that \AY{Option~(iii)} in Claim~A occurs.
This implies that either $v_4 v_3 w_3 a$ or $v_4 v_3 w_3 b$ is a $P_4$ in $G-D$ where the end-points are not adjacent in $G$, a contradiction.
So \AY{Option~(iii)} in Claim~A cannot occur, which implies that Option~(i) in Claim~A must occur.
 This implies that $V(C)$ is an R-deletable set, so we can reduce $G$ to $G-V(C)$ by Lemma~\ref{lem:Rdel}.
This proves Claim~D.\smallqed{}

We now return to the proof of Lemma~\ref{lem:noC4with3deg3}.
If $d(w_1)=3$ then let $s=1$ and otherwise let $s=2$. 
By Claim~D we note that $d(w_{3-s})=2$.

Let $G^*$ be obtained from $G-\{v_1,v_2,v_3,v_4,w_{3-s}\}$ by adding a vertex $x$ and the edge $w_s w_3$ and $w_s x$ and setting $\del{}(x)=0$ and $\add{}(x)=1$.
See Figure~\ref{fig:pic4} for an illustration.
We will now show that $G^*$ is a YES-instance if and only if $G$ is a YES-instance.

\begin{figure}[hbtp]
\begin{center}
\tikzstyle{vertexB}=[circle,draw, minimum size=14pt, scale=0.6, inner sep=0.5pt]
\tikzstyle{vertexR}=[circle,draw, color=red!100, minimum size=14pt, scale=0.6, inner sep=0.5pt]

\begin{tikzpicture}[scale=0.3]
 \node (v1) at (1,8) [vertexB] {$v_1$};
 \node (v2) at (9,8) [vertexB] {$v_2$};
 \node (v3) at (9,2) [vertexB] {$v_3$};
 \node (v4) at (1,2) [vertexB] {$v_4$};
 \node (w1) at (6,11) [vertexB] {$w_1$};
 \node (w2) at (14,11) [vertexB] {$w_2$};
 \node (w3) at (14,5) [vertexB] {$w_3$};

\draw [line width=0.03cm] (v1) to (v2);
\draw [line width=0.03cm] (v2) to (v3);
\draw [line width=0.03cm] (v3) to (v4);
\draw [line width=0.03cm] (v4) to (v1);

\draw [line width=0.03cm] (w1) to (v1);
\draw [line width=0.03cm] (w2) to (v2);
\draw [line width=0.03cm] (w3) to (v3);

\draw [line width=0.03cm] (w1) to (w2);

\draw [dashed, line width=0.03cm] (w2) to (17,11);
\draw [dashed, line width=0.03cm] (w3) to (17,5);

\draw (22,7) node {{\Large $\Rightarrow$}};
\end{tikzpicture}
\begin{tikzpicture}[scale=0.3]
\draw (-1,2) node {\mbox{ }};
\draw (14,11) node {\mbox{ }};
\draw (6,9.9) node {{\tiny $\del{}(x)=0$}};

 \node (x) at (6,11) [vertexB] {$x$};
 \node (ws) at (14,11) [vertexB] {$w_s$};
 \node (w3) at (14,5) [vertexB] {$w_3$};


\draw [line width=0.03cm] (x) to (ws);
\draw [line width=0.03cm] (ws) to (w3);

\draw [dashed, line width=0.03cm] (w2) to (17,11);
\draw [dashed, line width=0.03cm] (w3) to (17,5);

\draw (8,5) node {$G^*$};
\end{tikzpicture}
\caption{An illustration of $G^*$ in the proof of Lemma~\ref{lem:noC4with3deg3} when $s=2$.}\label{fig:pic4}
\end{center}
\end{figure}
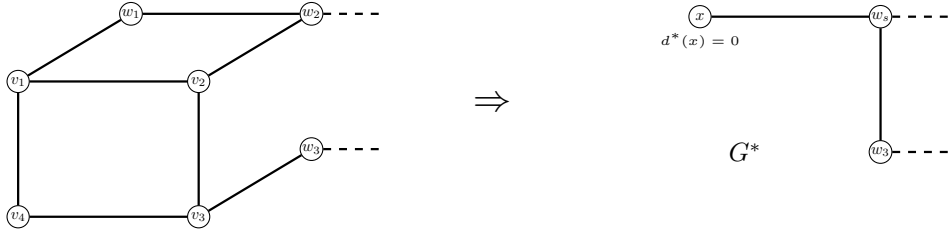

First assume that $G$ is a YES-instance. By Claim~D we note that Option~(i) or \AY{Option~(iii)} in Claim~A must occur.
If Option~(i) occurs then $v_i w_i \in D$ for all $i \in [3]$ in $G$. 
Let $D^* = D \cup \{w_s w_3 \} - \{v_1 w_1, v_2 w_2, v_3 w_3\}$. Now
$G-D$ and $G^* - D^*$ have exactly the same components except the $4$-cycle $v_1 v_2 v_3 v_4 v_1$ in $G-D$ and the component containing
$w_1 w_2$ have been replaced by a component containing $x w_s$ in $G^* - D^*$. However the component containing $w_1 w_2$ in $G-D$ is
isomorphic with the component containing $x w_s$ in $G^*-D^*$. So we get a solution for $G^*$. See Figure~\ref{fig:pic5} for an illustration.

\begin{figure}[hbtp]
\begin{center}
\tikzstyle{vertexB}=[circle,draw, minimum size=14pt, scale=0.6, inner sep=0.5pt]
\tikzstyle{vertexR}=[circle,draw, color=red!100, minimum size=14pt, scale=0.6, inner sep=0.5pt]

\begin{tikzpicture}[scale=0.3]
\draw (1,10) node {$G - D$};
 \node (v1) at (1,8) [vertexB] {$v_1$};
 \node (v2) at (9,8) [vertexB] {$v_2$};
 \node (v3) at (9,2) [vertexB] {$v_3$};
 \node (v4) at (1,2) [vertexB] {$v_4$};
 \node (w1) at (6,11) [vertexB] {$w_1$};
 \node (w2) at (14,11) [vertexB] {$w_2$};
 \node (w3) at (14,5) [vertexB] {$w_3$};

\draw [line width=0.09cm] (v1) to (v2);
\draw [line width=0.09cm] (v2) to (v3);
\draw [line width=0.09cm] (v3) to (v4);
\draw [line width=0.09cm] (v4) to (v1);

 \draw [dotted, line width=0.01cm] (w1) to (v1);
 \draw [dotted, line width=0.01cm] (w2) to (v2);
 \draw [dotted, line width=0.01cm] (w3) to (v3);

\draw [line width=0.09cm] (w1) to (w2);

\draw [dashed, line width=0.03cm] (w2) to (17,11);
\draw [dashed, line width=0.03cm] (w3) to (17,5);

\draw (22,7) node {{\Large $\Rightarrow$}};
\end{tikzpicture}
\begin{tikzpicture}[scale=0.3]
\draw (-1,2) node {\mbox{ }};
\draw (14,11) node {\mbox{ }};
\draw (6,9.9) node {{\tiny $\del{}(x)=0$}};

 \node (x) at (6,11) [vertexB] {$x$};
 \node (ws) at (14,11) [vertexB] {$w_s$};
 \node (w3) at (14,5) [vertexB] {$w_3$};


\draw [line width=0.09cm] (x) to (ws);
\draw [dotted, line width=0.01cm] (ws) to (w3);

\draw [dashed, line width=0.03cm] (w2) to (17,11);
\draw [dashed, line width=0.03cm] (w3) to (17,5);

\draw (8,5) node {$G^* - D^*$};
\end{tikzpicture}
\caption{An illustration of $G-D$ and $G^* - D^*$ when part~(i) of Claim~A occurs.}\label{fig:pic5}
\end{center}
\end{figure}
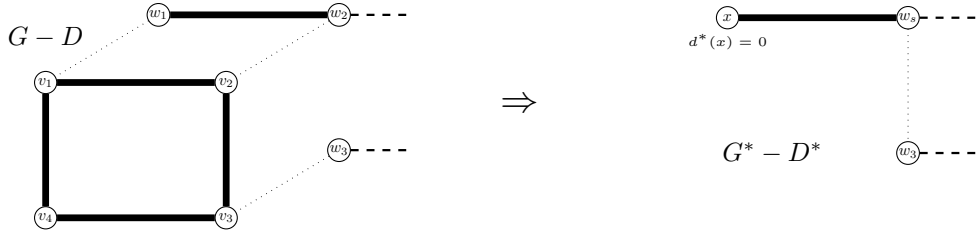

If \AY{Option~(iii)} occurs then $v_2 v_3, v_4 v_1 \in D$.
Let $D^* = D  - \{v_2 v_3, v_4 v_1\}$. Now
$G-D$ and $G^* - D^*$ have exactly the same components except the $4$-cycle $v_1 w_1 w_2 v_2 v_1$ and the path $v_4 v_3 w_3$ in $G-D$ 
have been replaced by the path $x w_s w_3$ in $G^* - D^*$. 
So we get a solution for $G^*$. See Figure~\ref{fig:pic6} for an illustration.

\begin{figure}[hbtp]
\begin{center}
\tikzstyle{vertexB}=[circle,draw, minimum size=14pt, scale=0.6, inner sep=0.5pt]
\tikzstyle{vertexR}=[circle,draw, color=red!100, minimum size=14pt, scale=0.6, inner sep=0.5pt]

\begin{tikzpicture}[scale=0.3]
\draw (1,10) node {$G - D$};

 \node (v1) at (1,8) [vertexB] {$v_1$};
 \node (v2) at (9,8) [vertexB] {$v_2$};
 \node (v3) at (9,2) [vertexB] {$v_3$};
 \node (v4) at (1,2) [vertexB] {$v_4$};
 \node (w1) at (6,11) [vertexB] {$w_1$};
 \node (w2) at (14,11) [vertexB] {$w_2$};
 \node (w3) at (14,5) [vertexB] {$w_3$};

\draw [line width=0.09cm] (v1) to (v2);
\draw [dotted, line width=0.01cm] (v2) to (v3);
\draw [line width=0.09cm] (v3) to (v4);
\draw [dotted, line width=0.01cm] (v4) to (v1);

 \draw [line width=0.09cm] (w1) to (v1);
 \draw [line width=0.09cm] (w2) to (v2);
 \draw [line width=0.09cm] (w3) to (v3);

\draw [line width=0.09cm] (w1) to (w2);

\draw [dotted, line width=0.01cm] (w2) to (17,11);
\draw [dotted, line width=0.01cm] (w3) to (17,5);

\draw (22,7) node {{\Large $\Rightarrow$}};
\end{tikzpicture}
\begin{tikzpicture}[scale=0.3]
\draw (-1,2) node {\mbox{ }};
\draw (14,11) node {\mbox{ }};
\draw (6,9.9) node {{\tiny $\del{}(x)=0$}};

 \node (x) at (6,11) [vertexB] {$x$};
 \node (ws) at (14,11) [vertexB] {$w_s$};
 \node (w3) at (14,5) [vertexB] {$w_3$};


\draw [line width=0.09cm] (x) to (ws);
\draw [line width=0.09cm] (ws) to (w3);

\draw [dotted, line width=0.01cm] (w2) to (17,11);
\draw [dotted, line width=0.01cm] (w3) to (17,5);

\draw (8,5) node {$G^* - D^*$};
\end{tikzpicture}
\caption{An illustration of $G-D$ and $G^* - D^*$ when part~(ii) of Claim~A occurs.}\label{fig:pic6}
\end{center}
\end{figure}
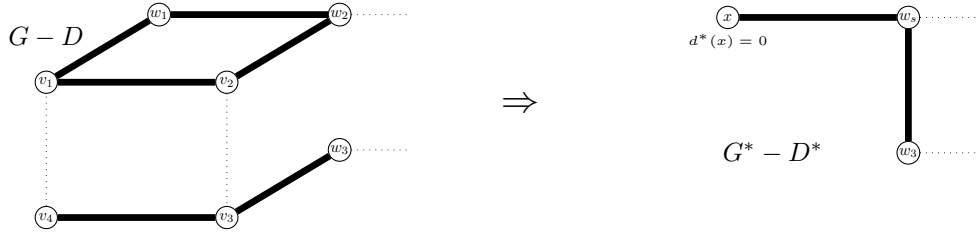

Conversely assume that $G^*$ is a YES-instance. We will now show that $G$ is a YES-instance. Let $D^*$ be the matching in $G^*$ such that
all components in $G^*-D^*$ are $P_1$, $P_2$, $P_3$ or $C_4$ (and the $P_3$'s and $C_4$'s can be made into cliques by adding a matching of addable edges).
As $\del{}(x)=0$ in $G^*$ we note that $x w_s \not\in D^*$. We now consider the case when $w_s w_3$ belongs to $D^*$ and then the case when it does not.

So first assume that $w_s w_3 \in D^*$. In this case let $D = D^* + \{v_1w_1, v_2 w_2, v_3 w_3 \} - \{w_s w_3\}$. 
Now $G-D$ and $G^* - D^*$ have exactly the same components except the $4$-cycle $v_1 v_4 w_4 w_1 v_1$ is in $G-D$ and the component containing
$x w_s$ have been replaced by an isomorphic  component containing $w_1 w_2$ in $G^* - D^*$. 
So we get a solution for $G$ (this corresponds to the reverse process to that illustrated in Figure~\ref{fig:pic5}).

Now assume that $w_s w_3 \not\in D^*$. This implies that any edges incident with $w_s$ and $w_3$, different from $w_s w_3$ and $x w_s$ must belong to \AY{$D^*$} (due
to the path $x w_s w_3$ in $G^* - D^*$). Let $D = D^* + \{v_2v_3, v_4 v_1 \}$.
Now $G-D$ and $G^* - D^*$ have exactly the same components except the \AY{path $x w_s w_3$ in $G^* - D^*$ has been replaced by the 
$4$-cycle $v_1 w_1 w_2 v_2 v_1$ and the path  $v_4 v_3 w_3$ in $G-D$.}
So we get a solution for $G$ (this corresponds to the reverse process to that illustrated in Figure~\ref{fig:pic6}). 

The above reduction completes the proof of the lemma.
\end{proof}

Recall that ${\cal G}_5$ denotes all $\{C_3,C_4\}$-free graphs of maximum degree at most 3.

\vspace{2mm}
 
\noindent {\bf Lemma \ref{lem:noC4}.}
{\em  $(1,1)$-{\sc Cluster Editing} can be reduced from ${\cal G}_4$ to ${\cal G}_5$ in polynomial time.}
\begin{proof}
Let $G \in {\cal G}_4$ \AY{be arbitrary with vertex-weights $(a^*,d^*)$.}
We may assume that $G$ is connected.
If there is no $4$-cycle in $G$ then we are done, so
let $C = v_1 v_2 v_3 v_4 v_1$ be any $4$-cycle in $G$.
As  $G \in {\cal G}_4$, we know that at most two of the vertices in $C$ have degree three.
If no vertex of $G$ has degree three then $V(G)=V(C)$ and we can solve our problem in constant time.
So we may assume that at least one vertex on $C$ has degree three.

\AY{If $\del{}(v_i)=0$ for any $i \in [4]$ then any solution $G-D+A$ to our instance must have all edges between $C$ and $V(G)-V(C)$ in $D$, which implies that
$V(C)$  is an R-deletable set. So by Lemma \ref{lem:Rdel}, we can reduce the instance in this case. Thus, we may assume that $\del{}(v_i)=1$ for all $i \in [4]$.}

We first consider the case when exactly one vertex on $C$ has degree three. Without loss of generality assume that $d(v_1)=3$ and $d(v_2)=d(v_3)=d(v_4)=2$. 
\AZ{Also, without loss of generality, assume that $a^*(v_2)\ge a^*(v_4)$ as otherwise we may rename the vertices in $C$ such that this is the case.}
Let $e$ be the edge incident with $v_1$ that does not belong to $C$.
Let $G^*$ be obtained from $G^*=G-\{v_3,v_4\}$ \AY{by letting $\del{}(v_2)=0$}.
Now $G^*$ is a YES-instance if and only if $G$ is a YES-instance, due to the following.
Deleting $e$ from $G$ corresponds to deleting $e$ from $G^*$ and vice versa.
Not deleting $e$ from $G$ corresponds to deleting two edges from $C$ in $G$ which corresponds to not deleting $e$ in $G^*$.
So $G$ is a YES-instance if and only if $G^*$ is a YES-instance.
We may therefore assume that exactly two vertices on $C$ have degree 3.

Let $x$ and $y$ be the two vertices on $C$ of degree 3. Let $x'$ be the neighbour of $x$ not on $C$ and let
$y'$ be the neighbour of $y$ not on $C$. Analogously to Claim~A in both Lemma~\ref{lem:noC4with4deg3} and Lemma~\ref{lem:noC4with3deg3} we note that 
if our instance has a solution and $D$ is the set of edges that are deleted from $G$ in the solution, then one of the following holds.

\begin{description}
\item[(i):] $x x', y y' \in D$  and no edge on $C$ belongs to $D$.
\item[(ii):] $x x', y y' \not\in D$  and $v_1 v_2, v_3 v_4 \in D$ and $v_2 v_3, v_4 v_1 \not\in D$.
\item[(iii):] $x x', y y' \not\in D$  and $v_2 v_3, v_4 v_1 \in D$ and $v_1 v_2, v_3 v_4 \not\in D$.
\end{description}

First consider the case when $x$ and $y$ are adjacent vertices on $C$. Then without loss of generality assume that $x=v_1$ and $y=v_4$. 
As $G$ is $C_3$-free we have $x' \not= y'$. First assume that $x' y' \in E(G)$. Then \AY{$V(G)=\{v_1, v_2, v_3, v_4, x', y'\}$} as 
$G \in {\cal G}_4$, which implies that the cycle $x' x y y' x'$ ($=x' v_1 v_4 y' x'$) has at most two vertices of degree three, which are $x$ and $y$.
So in this case we can solve our problem in constant time. Therefore, we may assume that $x' y'  \not\in E(G)$. This also implies that 
part~(ii) above is not possible, as if part~(ii) occurs then $x' v_1 v_4 y'$ is a $P_4$ in $G-D$ where $x' y' \not\in E(G)$, a contradiction.
So either part~(i) or part~(iii) occurs. This is illustrated in  Figure~\ref{fig:pic7}.

\begin{figure}[hbtp]
\begin{center}
\tikzstyle{vertexB}=[circle,draw, minimum size=14pt, scale=0.65, inner sep=0.5pt]
\tikzstyle{vertexR}=[circle,draw, color=red!100, minimum size=14pt, scale=0.6, inner sep=0.5pt]

\begin{tikzpicture}[scale=0.2]
\draw (8,0) node {$G$};

 \node (v1) at (6,8) [vertexB] {$x$};
 \node (v2) at (14,8) [vertexB] {$v_2$};
 \node (v3) at (14,2) [vertexB] {$v_3$};
 \node (v4) at (6,2) [vertexB] {$y$};
 \node (w1) at (1,8) [vertexB] {$x'$};
 \node (w4) at (1,2) [vertexB] {$y'$};

\draw [line width=0.03cm] (v1) to (v2);
\draw [line width=0.03cm] (v2) to (v3);
\draw [line width=0.03cm] (v3) to (v4);
\draw [line width=0.03cm] (v4) to (v1);


 \draw [line width=0.03cm] (v1) to (w1);
 \draw [line width=0.03cm] (v4) to (w4);
\draw (20,5) node {{\Large $\Rightarrow$}};
\end{tikzpicture} 
\begin{tikzpicture}[scale=0.2]
\draw (-4,5) node {\mbox{ }};
\draw (8,0) node {$G - D$};
 \node (v1) at (6,8) [vertexB] {$x$};
 \node (v2) at (14,8) [vertexB] {$v_2$};
 \node (v3) at (14,2) [vertexB] {$v_3$};
 \node (v4) at (6,2) [vertexB] {$y$};
 \node (w1) at (1,8) [vertexB] {$x'$};
 \node (w4) at (1,2) [vertexB] {$y'$};

\draw [line width=0.09cm] (v1) to (v2);
\draw [line width=0.09cm] (v2) to (v3);
\draw [line width=0.09cm] (v3) to (v4);
\draw [line width=0.09cm] (v4) to (v1);


 \draw [dotted, line width=0.01cm] (v1) to (w1);
 \draw [dotted, line width=0.01cm] (v4) to (w4);
\draw (20,5) node {{\small or}};
\end{tikzpicture}
\begin{tikzpicture}[scale=0.2]
\draw (-4,5) node {\mbox{ }};
\draw (8,0) node {$G - D$};
 \node (v1) at (6,8) [vertexB] {$x$};
 \node (v2) at (14,8) [vertexB] {$v_2$};
 \node (v3) at (14,2) [vertexB] {$v_3$};
 \node (v4) at (6,2) [vertexB] {$y$};
 \node (w1) at (1,8) [vertexB] {$x'$};
 \node (w4) at (1,2) [vertexB] {$y'$};

\draw [line width=0.09cm] (v1) to (v2);
\draw [dotted, line width=0.01cm] (v2) to (v3);
\draw [line width=0.09cm] (v3) to (v4);
\draw [dotted, line width=0.01cm] (v4) to (v1);

 \draw [line width=0.09cm] (v1) to (w1);
 \draw [line width=0.09cm] (v4) to (w4);
\end{tikzpicture}

\caption{Possible options for $G-D$ when $x$ and $y$ are adjacent.}\label{fig:pic7}
\end{center}
\end{figure}
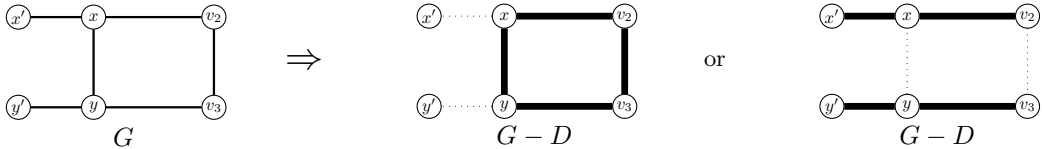

If $x$ and $y$ are not adjacent, then  without loss of generality assume that $x=v_1$ and $y=v_3$. If $x'=y'$ then 
\AY{$G$ contains a $K_{2,3}$-subgraph, which implies that our instance is a NO-instance by Lemma~\ref{lem:K23}.}
Thus, we may assume that $x' \not= y'$. 
Now  part~(i), part~(ii) and part~(iii) are illustrated in  Figure~\ref{fig:pic8}.

\begin{figure}[hbtp]
\begin{center}
\tikzstyle{vertexB}=[circle,draw, minimum size=14pt, scale=0.65, inner sep=0.5pt]
\tikzstyle{vertexR}=[circle,draw, color=red!100, minimum size=14pt, scale=0.6, inner sep=0.5pt]

\begin{tikzpicture}[scale=0.2]
\draw (10,-3) node {$G$};

 \node (v1) at (6,8) [vertexB] {$v_1$};
 \node (v2) at (14,8) [vertexB] {$v_2$};
 \node (v3) at (14,2) [vertexB] {$v_3$};
 \node (v4) at (6,2) [vertexB] {$v_3$};
 \node (w1) at (13,11) [vertexB] {$x'$};
 \node (w3) at (7,-1) [vertexB] {$y'$};

\draw [line width=0.03cm] (v1) to (v2);
\draw [line width=0.03cm] (v2) to (v3);
\draw [line width=0.03cm] (v3) to (v4);
\draw [line width=0.03cm] (v4) to (v1);


 \draw [line width=0.03cm] (v1) to (w1);
 \draw [line width=0.03cm] (v3) to (w3);
\draw (20,5) node {{\Large $\Rightarrow$}};
\end{tikzpicture}
\begin{tikzpicture}[scale=0.2]
\draw (0,5) node {\mbox{ }};
\draw (10,-3) node {$G$};

 \node (v1) at (6,8) [vertexB] {$v_1$};
 \node (v2) at (14,8) [vertexB] {$v_2$};
 \node (v3) at (14,2) [vertexB] {$v_3$};
 \node (v4) at (6,2) [vertexB] {$v_3$};
 \node (w1) at (13,11) [vertexB] {$x'$};
 \node (w3) at (7,-1) [vertexB] {$y'$};

\draw [line width=0.09cm] (v1) to (v2);
\draw [line width=0.09cm] (v2) to (v3);
\draw [line width=0.09cm] (v3) to (v4);
\draw [line width=0.09cm] (v4) to (v1);


 \draw [dotted, line width=0.01cm] (v1) to (w1);
 \draw [dotted, line width=0.01cm] (v3) to (w3);
\draw (20,5) node {or};
\end{tikzpicture}
\begin{tikzpicture}[scale=0.2]
\draw (0,5) node {\mbox{ }};
\draw (10,-3) node {$G$};

 \node (v1) at (6,8) [vertexB] {$v_1$};
 \node (v2) at (14,8) [vertexB] {$v_2$};
 \node (v3) at (14,2) [vertexB] {$v_3$};
 \node (v4) at (6,2) [vertexB] {$v_3$};
 \node (w1) at (13,11) [vertexB] {$x'$};
 \node (w3) at (7,-1) [vertexB] {$y'$};

\draw [line width=0.09cm] (v1) to (v2);
\draw [dotted, line width=0.01cm] (v2) to (v3);
\draw [line width=0.09cm] (v3) to (v4);
\draw [dotted, line width=0.01cm] (v4) to (v1);


 \draw [line width=0.09cm] (v1) to (w1);
 \draw [line width=0.09cm] (v3) to (w3);
\draw (20,5) node {or};
\end{tikzpicture}
\begin{tikzpicture}[scale=0.2]
\draw (0,5) node {\mbox{ }};
\draw (10,-3) node {$G$};

 \node (v1) at (6,8) [vertexB] {$v_1$};
 \node (v2) at (14,8) [vertexB] {$v_2$};
 \node (v3) at (14,2) [vertexB] {$v_3$};
 \node (v4) at (6,2) [vertexB] {$v_3$};
 \node (w1) at (13,11) [vertexB] {$x'$};
 \node (w3) at (7,-1) [vertexB] {$y'$};

\draw [dotted, line width=0.01cm] (v1) to (v2);
\draw [line width=0.09cm] (v2) to (v3);
\draw [dotted, line width=0.01cm] (v3) to (v4);
\draw [line width=0.09cm] (v4) to (v1);

 \draw [line width=0.09cm] (v1) to (w1);
 \draw [line width=0.09cm] (v3) to (w3);
\end{tikzpicture}

\caption{Possibble options for $G-D$ when $x$ and $y$ are not adjacent.}\label{fig:pic8}
\end{center}
\end{figure}

So whether $x$ and $y$ is adjacent or not, we will in $G-D$ either have both edges $xx'$ and $yy'$ in $D$ or neither in $D$. 
And if neither is in $D$ then each of $x'$ and $y'$ is the endpoint of a $P_3$ containing two vertices from $C$. So in both cases we will remove
$V(C)$ from $G$ and add the gadget $G(x,y)$, illustrated in Figure~\ref{fig:pic9}, to $G$ instead. Note that 
$V(G(x,y))=\{q_1, x, q_3, q_4, y, q_6 \}$ and $E(G(x,y))=\{q_1 x, x q_3, q_3 q_4, q_4 y, y q_5 \}$ and we add the edges $x x'$ and $y y'$.
Furthermore $\del{}(q_1)=\del{}(q_6)=0$ and all other $\del{}$-values and $\add{}$-values are one.

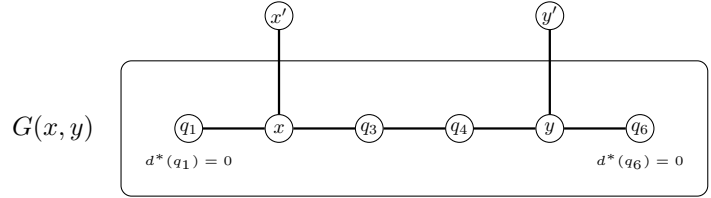
\begin{figure}[hbtp]
\begin{center}
\tikzstyle{vertexB}=[circle,draw, minimum size=14pt, scale=0.75, inner sep=0.5pt]
\tikzstyle{vertexR}=[circle,draw, color=red!100, minimum size=14pt, scale=0.6, inner sep=0.5pt]
\begin{tikzpicture}[scale=0.3]
\draw (0,6) node {\mbox{ }};
\draw (-4,6) node {$G(x,y)$};
\draw[rounded corners] (-1,3) rectangle (25,9) {};  

 \node (q1) at (2,6) [vertexB] {$q_1$};
 \node (q2) at (6,6) [vertexB] {$x$};
 \node (q3) at (10,6) [vertexB] {$q_3$};
 \node (q4) at (14,6) [vertexB] {$q_4$};
 \node (q5) at (18,6) [vertexB] {$y$};
 \node (q6) at (22,6) [vertexB] {$q_6$};

 \node (xp) at (6,11) [vertexB] {$x'$};
 \node (yp) at (18,11) [vertexB] {$y'$};

\draw (2,4.5) node {{\tiny $\del{}(q_1)=0$}};
\draw (22,4.5) node {{\tiny $\del{}(q_6)=0$}};

\draw [line width=0.03cm] (q1) to (q2);
\draw [line width=0.03cm] (q2) to (q3);
\draw [line width=0.03cm] (q3) to (q4);
\draw [line width=0.03cm] (q4) to (q5);
\draw [line width=0.03cm] (q5) to (q6);

\draw [line width=0.03cm] (q2) to (xp);
\draw [line width=0.03cm] (q5) to (yp);
\end{tikzpicture}

\caption{The gadget $G(x,y)$.}\label{fig:pic9}
\end{center}
\end{figure}

It is not difficult to check that if the modified graph, $G'$, is a YES-instance and $D'$ the matching that gets removed from $G'$, then we must have one of the two cases
illustrated in Figure~\ref{fig:pic10}. 

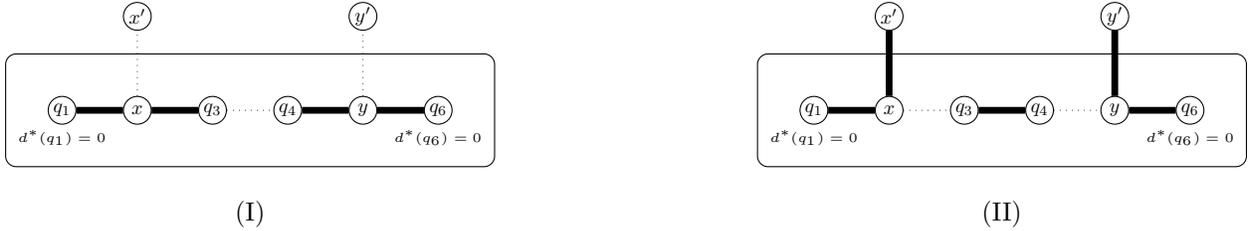
\begin{figure}[hbtp]
\begin{center}
\tikzstyle{vertexB}=[circle,draw, minimum size=14pt, scale=0.75, inner sep=0.5pt]
\tikzstyle{vertexR}=[circle,draw, color=red!100, minimum size=14pt, scale=0.6, inner sep=0.5pt]
\begin{tikzpicture}[scale=0.25]
\draw (0,6) node {\mbox{ }};
\draw[rounded corners] (-1,3) rectangle (25,9) {};  

 \node (q1) at (2,6) [vertexB] {$q_1$};
 \node (q2) at (6,6) [vertexB] {$x$};
 \node (q3) at (10,6) [vertexB] {$q_3$};
 \node (q4) at (14,6) [vertexB] {$q_4$};
 \node (q5) at (18,6) [vertexB] {$y$};
 \node (q6) at (22,6) [vertexB] {$q_6$};

 \node (xp) at (6,11) [vertexB] {$x'$};
 \node (yp) at (18,11) [vertexB] {$y'$};

\draw (2,4.5) node {{\tiny $\del{}(q_1)=0$}};
\draw (22,4.5) node {{\tiny $\del{}(q_6)=0$}};

\draw [line width=0.09cm] (q1) to (q2);
\draw [line width=0.09cm] (q2) to (q3);
\draw [dotted, line width=0.01cm] (q3) to (q4);
\draw [line width=0.09cm] (q4) to (q5);
\draw [line width=0.09cm] (q5) to (q6);

\draw [dotted, line width=0.01cm] (q2) to (xp);
\draw [dotted, line width=0.01cm] (q5) to (yp);
\draw (12,0.5) node {(I)};
\end{tikzpicture} \hfill
\begin{tikzpicture}[scale=0.25]
\draw (0,6) node {\mbox{ }};
\draw[rounded corners] (-1,3) rectangle (25,9) {};  

 \node (q1) at (2,6) [vertexB] {$q_1$};
 \node (q2) at (6,6) [vertexB] {$x$};
 \node (q3) at (10,6) [vertexB] {$q_3$};
 \node (q4) at (14,6) [vertexB] {$q_4$};
 \node (q5) at (18,6) [vertexB] {$y$};
 \node (q6) at (22,6) [vertexB] {$q_6$};

 \node (xp) at (6,11) [vertexB] {$x'$};
 \node (yp) at (18,11) [vertexB] {$y'$};

\draw (2,4.5) node {{\tiny $\del{}(q_1)=0$}};
\draw (22,4.5) node {{\tiny $\del{}(q_6)=0$}};

\draw [line width=0.09cm] (q1) to (q2);
\draw [dotted, line width=0.01cm] (q2) to (q3);
\draw [line width=0.09cm] (q3) to (q4);
\draw [dotted, line width=0.01cm] (q4) to (q5);
\draw [line width=0.09cm] (q5) to (q6);

\draw [line width=0.09cm] (q2) to (xp);
\draw [line width=0.09cm] (q5) to (yp);
\draw (12,0.5) node {(II)};
\end{tikzpicture} 
\caption{The two options, (I) and (II), for $G'-D'$.}\label{fig:pic10}
\end{center}
\end{figure}

So $G''$ has either  both edges $xx'$ and $yy'$ in $D'$ or neither in $D'$.
And if neither are in $D'$ then both $x'$ and $y'$ are endpoints of a $P_3$ containing two vertices from the gadget $G(x,y)$. 
This was exactly the same property we had in $G$, so $G$ is a YES-instance if and only if $G'$ is.

Furthermore, this reduction deletes a $4$-cycle. So even though the reduction adds an edge (that is, $|E(G')|=|E(G)|+1$) the reduction 
can be carried out at most as many times as there are $4$-cycles in $G$ which is at most a polynomial number.
Therefore repeatedly performing the above reduction creates a graph in ${\cal G}_5$ in polynomial time, as desired.
\end{proof}

Recall that ${\cal G}_6$ denotes all $\{C_3,C_4\}$-free graphs of maximum degree at most 3, such that the following holds.
\begin{itemize}
\item All vertices $v$ of degree at least 2 have $\del{}(v)=1$.
\item No vertex $v$ of degree 3 is adjacent to a vertex $w$ with $\add{}(w)=0$.
\item The vertices of degree 3 form an independent set.
\end{itemize}

\noindent{\bf Lemma \ref{lem:dZero}.}
{\em  $(1,1)$-{\sc Cluster Editing} can be reduced from ${\cal G}_5$ to ${\cal G}_6$ in polynomial time.}
\begin{proof}
Let $G \in {\cal G}_5$ \AY{be arbitrary with vertex-weights $(a^*,d^*)$.}
Note that if some vertex $v$ has $d(v)=3$ and $\del{}(v)=0$ then $(G,\add{},\del{})$ is a NO-instance, so we may assume that $\del{}(v)=1$ for all vertices, $v$, of degree 3.
If $d(v)=2$ and $\del{}(v)=0$ , then $(G,\add{},\del{})$ is a YES-instance if and only if both edges incident to $v$ have to remain in $G-D$ and therefore form a $P_3.$ Let $N(x)=\{u,v\},$ and note that $\{x,u,v\}$ is an R-deletable set. So by Lemma \ref{lem:Rdel}, we can reduce the instance. 

Assume that there exists a $uw \in E(G)$ with  $d(u)=3$ and $\add{}(w)=0$.
Assume that $G$ is a YES-instance and $G-D+A$ is a cluster graph. If $uw \not\in D$, then $w$ is the endpoint of a $P_3$ in $G-D$ which 
is a contradiction as $\add{}(w)=0$. So, we must have $uw \in D$. Since $uw \in D$, 
$\del{}(u)=\del{}(\AY{w})=1$ or $(G,a^*,d^*)$ is a NO-instance. Form a new instance with 
$G'=G- uw$ and $\del{}(u)=\del{}(\AY{w})=0$. Then $G$ is a YES-instance if and only if $G'$ is a YES-instance.
This reduction allows us to \AY{get} rid of all cases where a degree three vertex is adjacent to a vertex with $\add{}$-value zero.

Now assume that there exists a $uw \in E(G)$ with  $d(u)=d(w)=3$. 
Assume that $G$ is a YES-instance and $G-D+A$ is a cluster graph. If $uw \not\in D$, then we obtain a $P_4$ in $G-D$ containing the edge $uw$,
a contradiction. So, we must have $uw \in D$. As $d(u)=d(w)=3$, we have $d^*(u)=d^*(w)=1$.
Now we form a new instance with 
$G'=G- uw$ and $\del{}(u)=\del{}(v)=0$ as above and note 
that $G$ is a YES-instance if and only if $G'$ is a YES-instance. Thus, we may assume that $uw \not\in E(G)$ if $d(u)=d(w)=3$. 
This completes the proof.
\end{proof}

\section{Discussion}\label{sec:disc}

The main result of our paper concludes a complete dichotomy of complexity of $(a,d)$-{Cluster Editing}. 
We proved that $(1,1)$-{Cluster Editing} can be solved in polynomial time. 
Our proof consists of two stages: (i) providing a series of five polynomial-time reductions to $\{C_3,C_4\}$-free  graphs of maximum degree at most 3, and 
(ii)  designing a polynomial-time algorithm for solving $(1,1)$-{Cluster Editing} on $\{C_3,C_4\}$-free graphs of maximum degree at most 3. 
While Stage 2 is relatively short, Stage 1 is not.
Moreover, while our reduction from all graphs to $C_3$-free graphs of maximum degree at most 3 is not hard, getting rid of 4-cycles required 
a lengthy and non-trivial series of four polynomial-time reductions.  It would be interesting to see whether the four reductions can be replaced by a shorter and simpler series of reductions.

\paragraph{Acknowledgement} We are thankful to the referees for several useful comments and suggestions.
 


\begin{thebibliography}{10}

\bibitem{Abu-Khzam2017}
Faisal~N. Abu-Khzam.
\newblock On the complexity of multi-parameterized cluster editing.
\newblock {\em Journal of Discrete Algorithms}, 45:26--34, 2017.

\bibitem{Barr0AC19}
Joseph~R. Barr, Peter Shaw, Faisal~N. Abu{-}Khzam, and Jikang Chen.
\newblock Combinatorial text classification: the effect of multi-parameterized
  correlation clustering.
\newblock In {\em First International Conference on Graph Computing, {GC} 2019,
  Laguna Hills, CA, USA, September 25-27, 2019}, pages 29--36. {IEEE}, 2019.

\bibitem{BarrSATY20}
J.R. Barr, P.~Shaw, F.N. Abu{-}Khzam, T.~Thatcher, and S.~Yu.
\newblock Vulnerability rating of source code with token embedding and
  combinatorial algorithms.
\newblock {\em Int. J. Semantic Comput.}, 14(4):501--516, 2020.

\bibitem{Bocker12}
Sebastian B{\"{o}}cker.
\newblock A golden ratio parameterized algorithm for cluster editing.
\newblock {\em J. Discrete Algorithms}, 16:79--89, 2012.

\bibitem{BockerBBT09}
Sebastian B{\"{o}}cker, Sebastian Briesemeister, Quang Bao~Anh Bui, and Anke
  Tru{\ss}.
\newblock Going weighted: Parameterized algorithms for cluster editing.
\newblock {\em Theor. Comput. Sci.}, 410(52):5467--5480, 2009.

\bibitem{BockerBK11}
Sebastian B{\"{o}}cker, Sebastian Briesemeister, and Gunnar~W. Klau.
\newblock Exact algorithms for cluster editing: Evaluation and experiments.
\newblock {\em Algorithmica}, 60(2):316--334, 2011.

\bibitem{Cai96}
Liming Cai.
\newblock Fixed-parameter tractability of graph modification problems for
  hereditary properties.
\newblock {\em Inf. Process. Lett.}, 58(4):171--176, 1996.

\bibitem{CaoC12}
Yixin Cao and Jianer Chen.
\newblock Cluster editing: Kernelization based on edge cuts.
\newblock {\em Algorithmica}, 64(1):152--169, 2012.

\bibitem{ChenM12}
Jianer Chen and Jie Meng.
\newblock A $2k$ kernel for the cluster editing problem.
\newblock {\em J. Comput. Syst. Sci.}, 78(1):211--220, 2012.

\bibitem{GrammGHN05}
J.~Gramm, J.~Guo, F.~H{\"u}ffner, and R.~Niedermeier.
\newblock Graph-modeled data clustering: exact algorithms for clique
  generation.
\newblock {\em Theory Comput. Syst.}, 38(4):373--392, 2005.

\bibitem{Guo09}
J.~Guo.
\newblock A more effective linear kernelization for cluster editing.
\newblock {\em Theor. Comput. Sci.}, 410:718--726, 2009.

\bibitem{KomusiewiczU2012}
Christian Komusiewicz and Johannes Uhlmann.
\newblock Cluster editing with locally bounded modifications.
\newblock {\em Discrete Applied Mathematics}, 160(15):2259--2270, 2012.

\bibitem{KrivanekM86}
M.~Krivanek and J.~Moravek.
\newblock {NP}-hard problems in hierarchical-tree clustering.
\newblock {\em Acta Inform.}, 23(3):311--323, 1986.

\bibitem{ShamirST04}
Ron Shamir, Roded Sharan, and Dekel Tsur.
\newblock Cluster graph modification problems.
\newblock {\em Discret. Appl. Math.}, 144(1-2):173--182, 2004.

\bibitem{ShawBA22}
Peter Shaw, Joseph~R. Barr, and Faisal~N. Abu{-}Khzam.
\newblock Anomaly detection via correlation clustering.
\newblock In {\em 16th {IEEE} International Conference on Semantic Computing,
  {ICSC} 2022, Laguna Hills, CA, USA, January 26-28, 2022}, pages 307--313.
  {IEEE}, 2022.

\end{thebibliography}
\end{document}